\documentclass[a4paper, USenglish, cleveref, autoref, thm-restate]{lipics-v2019}

\graphicspath{{./images/}}

\usepackage{graphicx,amssymb,amsmath}

\usepackage{amsfonts}
\usepackage{booktabs}
\usepackage{cite}
\usepackage{color}
\usepackage{enumerate}
\usepackage{float}
\usepackage{hyperref}
\usepackage{mathrsfs}
\usepackage{tabularx}
\usepackage[ruled]{algorithm2e}
\SetKwInput{KwResult}{Function}

\def\calT{\mathcal{T}}
\def\e{\overline{e}}

\title{Algorithms for Diameters of Unicycle Graphs and Diameter-Optimally Augmenting Trees\footnote{A preliminary version of this paper appears in the Proceedings of the 15th International Conference and Workshops on Algorithms and Computation (WALCOM 2021).}} 

\titlerunning{Diameters of Unicycle Graphs and Diameter-Optimally Augmenting Trees} 

\author{Haitao Wang}{Department of Computer Science,
Utah State University, Logan, UT 84322, USA.}{haitao.wang@usu.edu}{}{}

\author{Yiming Zhao\footnote{Corresponding author.}}{Department of Computer Science,
Utah State University, Logan, UT 84322, USA.}{yiming.zhao@usu.edu}{}{}

\authorrunning{H. Wang and Y. Zhao} 

\Copyright{Haitao Wang and Yiming Zhao} 



\category{} 

\relatedversion{} 

\supplement{}

\funding{This research was supported in part by NSF under Grant CCF-2005323.}

\acknowledgements{}

\nolinenumbers 

\hideLIPIcs  

\EventEditors{John Q. Open and Joan R. Access}
\EventNoEds{2}
\EventLongTitle{42nd Conference on Very Important Topics (CVIT 2016)}
\EventShortTitle{CVIT 2016}
\EventAcronym{CVIT}
\EventYear{2016}
\EventDate{December 24--27, 2016}
\EventLocation{Little Whinging, United Kingdom}
\EventLogo{}
\SeriesVolume{42}
\ArticleNo{23}

\begin{document}

\maketitle

\begin{abstract}
We consider the problem of computing the diameter of a unicycle graph (i.e., a graph with a unique cycle). We present an $O(n)$ time algorithm for the problem, where $n$ is the number of vertices of the graph. This improves the previous best $O(n\log n)$ time solution [Oh and Ahn, ISAAC 2016]. Using this algorithm as a subroutine, we solve the problem of adding a shortcut to a tree so that the diameter of the new graph (which is a unicycle graph) is minimized; our algorithm takes $O(n^2\log n)$ time and $O(n)$ space. The previous best algorithms solve the problem in $O(n^2\log^3 n)$ time and $O(n)$ space [Oh and Ahn, ISAAC 2016], or in $O(n^2)$ time and $O(n^2)$ space [Bil{\`o}, ISAAC 2018].
\end{abstract}

\section{Introduction}
\label{sec:Introduction}

Let $G$ be a graph of $n$ vertices such that each edge has a positive length. A {\em shortest path} connecting two vertices $s$ and $t$ in $G$ is a path of minimum total edge length; the length of the shortest path is also called the {\em distance} between $s$ and $t$ in $G$. The {\em diameter} of $G$ is the maximum distance between all pairs of vertices of $G$. $G$ is a {\em unicycle graph} if it has only one cycle, i.e., $G$ is a tree plus an additional edge.

We consider the problem of computing the diameter of a unicycle graph $G$. Previously, Oh and Ahn~\cite{ref:OhA16} solved the problem in $O(n\log n)$ time, where $n$ is the number of vertices of $G$. We present an improved algorithm of $O(n)$ time.
Using our new algorithm, we also solve the \emph{diameter-optimally augmenting tree} (DOAT for short) problem, defined as follows.

Let $T$ be a tree of $n$ vertices such that each edge has a positive length. We want to add a new edge (called {\em shortcut}) to $T$ such that the new graph (which is a unicycle graph) has the minimum diameter. We assume that there is an oracle that returns the length of any given shortcut in $O(1)$ time. Previously, Oh and Ahn~\cite{ref:OhA16} solved the problem in $O(n^2\log^3 n)$ time and $O(n)$ space, and Bil{\`o}~\cite{ref:BiloAl18} reduced the time to $O(n^2)$ but the space increases to $O(n^2)$. As observed by Oh and Ahn~\cite{ref:OhA16}, the problem has an $\Omega(n^2)$ lower bound on the running time as all $\Theta(n^2)$ possible shortcuts have to be checked in order to find an optimal shortcut. Hence, Bil{\`o}'s algorithm is time-optimal.
In this paper, we propose an algorithm with a better time and space trade-off, and our algorithm uses $O(n^2\log n)$ time and $O(n)$ space.

\subsection{Related work}
\label{subsec:RelatedWork}

The diameter is an important measure of graphs and computing it is one of the most fundamental algorithmic graph problems. For general graphs or even planar graphs, the only known way to compute the diameter is to first solve the all-pair-shortest-path problem (i.e., compute the distances of all pairs of vertices of the graph), which inherently takes $\Omega(n^2)$ time, e.g.,~\cite{ref:FedericksonFa87,ref:WilliamsFa18}.
Better algorithms exist for special graphs. For example, the diameter of a tree can be computed in linear time, e.g., by first computing its center~\cite{ref:MegiddoLi83}.
If $G$ is an outerplanar graph and all edges have the same length, its diameter can be computed in linear time~\cite{ref:FarleyCo80}. The diameter of interval graphs (with equal edge lengths) can also be computed in linear time~\cite{ref:Olariu90}. Our result adds the unicycle graph (with different edge lengths) to the linear-time solvable graph category.


The DOAT problem and many of its variations enjoy an increasing interest in the research community. If the tree $T$ is embedded in a metric space (so that the triangle inequality holds for edge lengths), Gro{\ss}e et al. \cite{ref:GrobeFa16} first solved the problem in $O(n^2 \log n)$ time. Bil{\`o} \cite{ref:BiloAl18} later gave an $O(n \log n)$ time and $O(n)$ space algorithm, and another $(1 + \epsilon)$-approximation algorithm of $O(n + \frac{1}{\epsilon} \log \frac{1}{\epsilon})$ time and $O(n + \frac{1}{\epsilon})$ space for any $\epsilon>0$. A special case where $T$ is a path embedded in a metric space was first studied by Gro{\ss}e et al.~\cite{ref:GrobeFa15}, who gave an $O(n\log^3 n)$ time algorithm,
and the algorithm was later improved to $O(n\log n)$ time by Wang~\cite{ref:WangAn18}. Hence,  Bil{\`o}'s work~\cite{ref:BiloAl18} generalizes Wang's result~\cite{ref:WangAn18} to trees.

A variant of the DOAT problem which aims to minimize the \emph{continuous diameter}, i.e., the diameter of $T$ is measured with respect to all the points of the tree (including the points in the interior of the edges), has also been studied. If $T$ is a path embedded in the Euclidean plane, De Carufel et al.~\cite{ref:DeCarufelMi16} solved the problem in $O(n)$ time. If $T$ is a tree embedded in a metric space, De Carufel et al.~\cite{ref:DeCarufelMi17} gave an $O(n\log n)$ time algorithm.
If $T$ is a general tree, Oh and Ahn \cite{ref:OhA16} solved the problem in $O(n^2 \log^3 n)$ time and $O(n)$ space.

The DOAT problem is to minimize the diameter. The problem of minimizing the radius was also considered. For the case where $T$ is a path embedded in a metric space, Johnson and Wang~\cite{ref:JohnsonAL19} presented a linear time algorithm which adds a shortcut to $T$ so that the radius of the resulting graph is minimized. The radius considered in~\cite{ref:JohnsonAL19} is defined with respect to all points of $T$, not just the vertices. Wang and Zhao \cite{ref:WangAl20} studied the same problem with radius defined with respect to only the vertices, and they
gave a linear time algorithm.

The more general problem in which one wants to add $k$ shortcuts to a graph to minimize the diameter is NP-hard~\cite{ref:SchooneDi97} and some variations are even W[2]-hard~\cite{ref:FratiAu15, ref:GaoTh13}.
Approximation algorithms have been proposed~\cite{ref:BiloIm12, ref:ChepoiAu02, ref:DemaineMi10, ref:FratiAu15, ref:LiOn92}. The upper and lower bounds on the values of diameters of certain augmented graphs were also studied, e.g., \cite{ref:AlonDe00, ref:ChungDi84, ref:ToshimasaAu13}.
Bae et al. \cite{ref:BaeSh19} considered the problem of adding $k$ shortcuts to
a circle in the plane to minimize the diameter of the resulting graph.

\subsection{Our approach}
\label{subsec:OurResult}

To compute the diameter of a unicycle graph $G$, Oh and Ahn~\cite{ref:OhA16} reduces the problem to a geometric problem and then uses a one-dimensional range tree to solve the problem. We take a completely different approach. Let $C$ be the unique cycle of $G$. We define certain ``domination'' relations on the vertices of $C$ so that if a vertex $v$ is dominated by another vertex then $v$ is not important to the diameter. We then present a pruning algorithm to find all undominated vertices (and thus those dominated vertices are ``pruned''); it turns out that finding the diameter among the undominated vertices is fairly easy. In this way, we compute the diameter of $G$ in linear time.

For the DOAT problem on a tree $T$, Oh and Ahn~\cite{ref:OhA16} considered all possible shortcuts of $T$ by following an Euler tour of $T$; they used the aforementioned one-dimensional range tree to update the diameter for the next shortcut. Bil{\`o}'s method~\cite{ref:BiloAl18} is to transform the problem to adding a shortcut to a path whose edge lengths satisfy a property  similar in spirit to the triangle inequality (called graph-triangle inequality) and then the problem on $P$ can be solved by applying the $O(n\log n)$ time algorithm for trees in metric space~\cite{ref:BiloAl18}. Unfortunately, the problem transformation algorithm relies on using $O(n^2)$ space to store the lengths of all possible $\Theta(n^2)$ shortcuts of $T$. The algorithm has to consider all these $\Theta(n^2)$ shortcut lengths in a global manner and thus it inherently uses $\Omega(n^2)$ space. Note that Bil{\`o}'s method~\cite{ref:BiloAl18} does not need an algorithm for computing the diameter of a unicycle graph.

We propose a novel approach. We first compute a diametral path $P$ of $T$. Then we reduce the DOAT problem on $T$ to finding a shortcut for $P$. To this end, we consider vertices of $P$ individually. For each vertex $v_i$ of $P$, we want to find an optimal shortcut with the restriction that it must connect $v_i$, dubbed a {\em $v_i$-shortcut}. For this, we define a ``domination'' relation on all $v_i$-shortcuts and we show that those shortcuts dominated by others are not important. We then design a pruning algorithm to find all shortcuts that are not dominated by others; most importantly, these undominated shortcuts have certain monotonicity properties that allow us to perform binary search to find an optimal $v_i$-shortcut by using our diameter algorithm for unicycle graphs as a subroutine. With these effort, we find an optimal $v_i$-shortcut in $O(n\log n)$ time and $O(n)$ space. The space can be reused for computing optimal $v_i$-shortcuts of other vertices of $P$. In this way, the total time of the algorithm is $O(n^2\log n)$ and the space is $O(n)$.

\subparagraph{Outline.}
In the following, we present our algorithm for computing the diameter of a unicycle graph in Section~\ref{sec:ComputingTheDiameterOfAUnicylicGraph}. Section \ref{sec:DOAPForATree} is concerned with the DOAT problem.


\section{Computing the Diameter of Unicycle Graphs}
\label{sec:ComputingTheDiameterOfAUnicylicGraph}

In this section, we present our linear time algorithm for computing the diameter of unicycle graphs.

For a subgraph $G'$ of a graph $G$ and two vertices $u$ and $v$ from $G'$, we use $\pi_{G'}(u,v)$ to denote a shortest path from $u$ to $v$ in $G'$ and use $d_{G'}(u,v)$ to denote the length of the path.
We use $\Delta(G)$ to denote the diameter of $G$. A pair of vertices $(u,v)$ is called a {\em diametral pair} and $\pi_G(u,v)$ is called a {\em diametral path} if $d_G(u,v)=\Delta(G)$.

In the following, let $G$ be a unicycle graph of $n$ vertices. Our goal is to compute the diameter $\Delta(G)$ (along with a diametral pair). Let $C$ denote the unique cycle of $G$.

\subsection{Observations}

Removing all edges of $C$ (while keeping its vertices) from $G$ results in several connected components of $G$. Each component is a tree that contains a vertex $v$ of $C$; we use $T(v)$ to denote the tree. Let $v_1,v_2,\ldots,v_m$ be the vertices ordered clockwise on $C$. Let $\calT(G)=\{T(v_i)\ |\ 1\leq i\leq m\}$. Note that the sets of vertices of all trees of $\calT(G)$ form a partition of the vertex set of $G$.

Consider a diametral pair $(u^*,v^*)$ of $G$. There are two cases: (1) both $u^*$ and $v^*$ are in the same tree of $\calT(G)$; (2) $u^*$ and $v^*$ are in two different trees of $\calT(G)$. To handle the first case, we compute the diameter of each tree of $\calT(G)$, which can be done in linear time. Computing the diameters for all trees takes $O(n)$ time. The longest diameter of these trees is the diameter of $G$. In the following, we focus on the second case.

Suppose $T(v_i)$ contains $u^*$ and $T(v_j)$ contains $v^*$ for $i\neq j$. Observe that the diametral path $\pi_G(u^*,v^*)$ is the concatenation of the following three paths: $\pi_{T(v_i)}(u^*,v_i)$, $\pi_C(v_i,v_j)$, and $\pi_{T(v_j)}(v_j,v^*)$. Further, $u^*$ is the farthest vertex in $T(v_i)$ from $v_i$; the same holds for $v^*$ and $T(v_j)$. On the basis of these observations, we introduce some concepts as follows.



For each vertex $v_i\in C$, we define
a {\em weight} $w(v_i)$ as the length of the path from $v_i$ to its farthest vertex in $T(v_i)$. 
The weights for all vertices on $C$ can be computed in total $O(n)$ time.
With this definition in hand, we have $\Delta(G)=\max_{1\leq i< j\leq m}(w(v_i) + d_C(v_i,v_j) + w(v_j))$. We say that $(v_i,v_j)$ is a {\em vertex-weighted diametral pair} of $C$ if $T(v_i)$ contains $u^*$ and $T(v_j)$ contains $v^*$ for a diametral pair $(u^*,v^*)$ of $G$. To compute $\Delta(G)$, it suffices to find a vertex-weighted diameter pair of $C$.

We introduce a domination relation for vertices on $C$.
\begin{definition}
    \label{def:domination-cycle}
    For two vertices $v_{i}, v_{j} \in C$, we say that $v_{i}$ {\em dominates} $v_{j}$ if $w(v_i) > w(v_j) + d_{C}(v_{i}, v_{j})$.
\end{definition}

The following lemma shows that if a vertex is dominated by another vertex, then it is not ``important''.


\begin{lemma}
    \label{lemma:1}
    For two vertices $v_i$ and $v_j$ of $C$, if $v_{i}$ dominates $v_{j}$, then $v_{j}$ cannot be in any  vertex-weighted diametral pair of $C$ unless $(v_i,v_j)$ is such a pair.
\end{lemma}
\begin{proof}
    \label{proof:1}
    We assume that $(v_i,v_j)$ is not a vertex-weighted diametral pair.
        Assume to the contrary that $(v_{k}, v_{j})$ is a vertex-weighted diametral pair of $C$. Then, $k\neq i$ and $\Delta(G) = w(v_{k}) + d_{C}(v_{k}, v_{j}) + w(v_{j})$.
        Note that $d_{C}(v_{k}, v_{j}) \leq d_{C}(v_{k}, v_{i}) + d_{C}(v_{i}, v_{j})$ holds.
        Since $v_{i}$ dominates $v_{j}$, we have $w(v_{i}) > w(v_{j}) + d_{C}(v_{i}, v_{j})$.
        Consequently, we can derive
    \begin{align*}
        \Delta(G) &= w(v_{k}) + d_{C}(v_{k}, v_{j}) + w(v_{j}) \\
        &\leq w(v_{k}) + d_{C}(v_{k}, v_{i}) + d_{C}(v_{i}, v_{j}) + w(v_{j})\\
        &< w(v_{k}) + d_{C}(v_{k}, v_{i}) + w(v_{i}).
    \end{align*}
      But this contradicts with the definition of $\Delta(G)$. The lemma thus follows.
\end{proof}


\subsection{A pruning algorithm}

In the sequel, we describe a linear time {\em pruning algorithm} to find all vertices of $C$ that are dominated by other vertices (and thus those dominated vertices are ``pruned''). As will be seen later, the diameter can be easily found after these vertices are pruned.

\begin{figure}[t]
    \centering
    \includegraphics[width=1.8in]{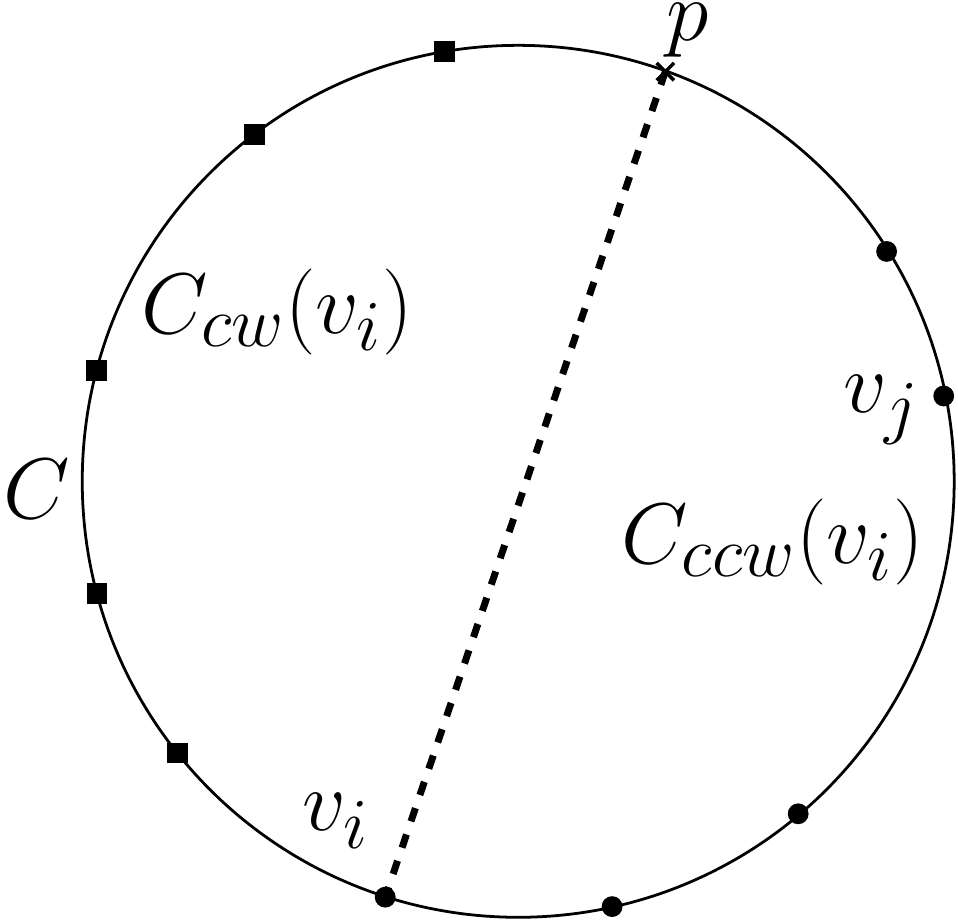}
    \caption{Illustrating the definitions of $C_{ccw}(v_i)$ (the disks except $v_i$) and $C_{cw}(v_i)$ (the squares). We assume that $p$ is a point on $C$ that together with $v_i$ partitions $C$ into two half-cycles of equal length.}
    \label{fig:VertexWeightedCycleRemoved}
\end{figure}

Let $|C|$ denote the sum of the lengths of all edges of $C$.
For any vertex $v_i$ of $C$, define $C_{ccw}(v_i)$ as the set of vertices $v_j$ of $C$ such that the path from $v_i$ to $v_j$ counterclockwise along $C$ has length at most $|C|/2$ (e.g., see Fig~\ref{fig:VertexWeightedCycleRemoved}); define $C_{cw}(v_i)$ as the set of vertices of $C$ not in $C_{ccw}(v_i)$. We assume that $v_i$ is in neither subset.

\begin{lemma}\label{lem:preprocess}
With $O(n)$ time preprocessing, given any two vertices $v_i$ and $v_j$ of $C$, we can do the following in $O(1)$ time: (1) compute $d_C(v_i,v_j)$; (2) determine whether $v_j$ is in $C_{ccw}(v_i)$; (3) determine whether $v_i$ and $v_j$ dominate each other.
\end{lemma}
\begin{proof}
We first compute the weight $w(v_i)$ for all vertices $v_i\in C$. This can be done in $O(n)$ time. Then, we compute the length $|C|$. Next, by scanning the vertices $v_1,v_2,\ldots, v_m$ on $C$, we compute an array $A[1,\ldots,m]$ with $A[i]$ equal to the length of the path from $v_1$ to $v_i$ clockwise along $C$. Hence, for any $1\leq i<j\leq m$, $A[j]-A[i]$ is the length of the path from $v_i$ to $v_j$ clockwise along $C$ and $|C|-(A[j]-A[i]))$ is the length of the path from $v_i$ to $v_j$ counterclockwise along $C$. Note that $d_C(v_i,v_j)=\min\{A[j]-A[i],|C|-(A[j]-A[i])\}$.

Consider any two vertices $v_i$ and $v_j$ of $C$. Without loss of generality, we assume $i<j$. By comparing $A[j]-A[i]$ with $|C|/2$, we can determine whether $v_j$ is in $C_{ccw}(v_i)$ in $O(1)$ time. As $w(v_i)$ and $w(v_j)$ are both available, whether $v_i$ and $v_j$ dominate each other can be determined in $O(1)$ time.
\end{proof}

With Lemma~\ref{lem:preprocess} in hand, starting from $v_1$, our pruning algorithm processes the vertices of $C$ from $v_1$ to $v_m$ in order (see Algorithm~\ref{algorithm:UnicyclicGraph} for the pseudocode). The algorithm maintains a stack $S$, which is empty initially. Consider a vertex $v_i$. If $S=\emptyset$, then we push $v_i$ into $S$. Otherwise, let $v$ be the vertex at the top of $S$. If $v$ is not in $C_{ccw}(v_i)$, then we also push $v_i$ into $S$. Otherwise, we check whether $v$ and $v_i$ dominate each other. If they do not dominate each other, then we push $v_i$ into $S$.
Otherwise, if $v_i$ dominates $v$, we pop $v$ out of $S$, and then we continue to pop the new top element $v$ of $S$ out as long as the following three conditions are all satisfied: (1) $S\neq \emptyset$; (2) $v\in C_{ccw}(v_i)$; (3) $v_i$ dominates $v$. Once one of the three conditions is not satisfied, we push $v_i$ into $S$.

After $v_m$ is processed, the first stage of the pruning algorithm is over. In the second stage, we process the vertices in the stack $S$ in a bottom-up manner until a vertex not in $C_{cw}(v_1)$; the processing of a vertex is done in the same way as above (the vertex should be removed from $S$ first). Specifically, let $v_i$ be the vertex at the bottom of $S$. If $v_i$ is not in $C_{cw}(v_1)$, then we stop the algorithm and return the vertices in the current stack $S$. Otherwise, we remove $v_i$ from $S$ and then apply the same processing algorithm as above in the first stage (i.e., begin with checking whether $S$ is empty).

Intuitively, the first stage of the algorithm does a ``full-cycle'' scan on $C$ while the second stage does a ``half-cycle'' scan (i.e., the half-cycle clockwise from $v_1$). With Lemma~\ref{lem:preprocess}, the algorithm can be implemented in $O(n)$ time. The following lemma establishes the correctness of the algorithm.

   \begin{algorithm}[htbp]
    \DontPrintSemicolon
    \caption{The pruning algorithm}
    \label{algorithm:UnicyclicGraph}
    \SetKwFunction{FMain}{Main}
    \SetKwFunction{FProcessVertex}{ProcessVertex}

    \SetKwProg{Fn}{Function}{:}{end}

    \Fn{\FProcessVertex{$v_i$, $S$}}
    {
        \If{$S == \emptyset$}
        {
            $S.\mathbf{push}(v_i)$ \tcp{Push vertex $v_i$ into the stack}
        }
        \Else{
        $v\leftarrow S.\mathbf{top}()$ \tcp{$S.\mathbf{top}()$ is the top element of the stack}
        \If{$v \notin C_{ccw}(v_i)$}
        {
            $S.\mathbf{push}(v_i)$\;
        }
        \ElseIf{$v_i$ and $v$ do not dominate each other}
        {
            $S.\mathbf{push}(v_i)$\;
        }
        \ElseIf{$v_i$ dominates $v$}
        {
            $S.\mathbf{pop}()$ \tcp{Pop the top element out of $S$}
            \While{$S \neq \emptyset$ $\mathbf{and}$ $S.\mathbf{top}() \in C_{ccw}(v_i)$ $\mathbf{and}$ $v_i$ dominates $S.\mathbf{top}()$}
            {
                $S.\mathbf{pop}()$
            }
            $S.\mathbf{push}(v_i)$\;
        }
        }

    }

    \Fn{\FMain{$S = \emptyset$, $C$}}
    {
        \tcp{the full-cycle scan}
        $S=\emptyset$\;
        \For {$i = 1, 2, ..., m$}
        {
            ProcessVertex($v_{i}$, $S$) \tcp{Call ProcessVertex function on $v_i$}
        }
        \tcp{the half-cycle scan}
        $v_i\leftarrow $ the bottom element of $S$\;
        \While{$v_i\in C_{cw}(v_1)$}
        {
           Remove $v_i$ from $S$\;
           ProcessVertex($v_{i}$, $S$)\;
           $v_i\leftarrow $ the bottom element of $S$\;
        }
        \Return $S$\;
    }

    \end{algorithm}

\begin{lemma}\label{lem:pruning}
Let $S$ be the stack after the algorithm is over.
\begin{enumerate}
    \item Each vertex of $C$ that is not in $S$ is dominated by a vertex in $S$.
    \item No two vertices of $S$ dominate each other.
\end{enumerate}
\end{lemma}
\begin{proof}
We first proves an observation about
two transitive properties of the domination relation, which will be used to prove the lemma.
\subparagraph{\em Observation.}
{\em    Let $v_{i}, v_{j}, v_{k}$ be any three vertices of $C$.
    \begin{enumerate}
        \item If $v_{i}$ dominates $v_{j}$ and $v_{j}$ dominates $v_{k}$, then $v_{i}$ dominates $v_{k}$.
        \item If $v_{i}$ and $v_{j}$ do not dominate each other, $v_{j}$ and $v_{k}$ do not dominate each other, and $d_{C}(v_{i}, v_{k}) = d_{C}(v_{i}, v_{j}) + d_{C}(v_{j}, v_{k})$, then $v_{i}$ and $v_{k}$ do not dominate each other.
    \end{enumerate} }
\subparagraph{Proof of observation.}
    \begin{enumerate}
        \item Since $v_{i}$ dominates $v_{j}$ and $v_{j}$ dominates $v_{k}$, we have $w(v_{i}) > w(v_{j}) + d_{C}(v_{i}, v_{j})$ and $w(v_{j}) > w(v_{k}) + d_{C}(v_{j}, v_{k})$. Thus, $w(v_{i}) > w(v_{k}) + d_{C}(v_{j}, v_{k}) + d_{C}(v_{i}, v_{j}) \geq w(v_{k}) + d_{C}(v_{i}, v_{k})$. Hence, $v_{i}$ dominates $v_{k}$.

        \item
        As $v_{i}$ and $v_{j}$ do not dominate each other, $w(v_{i}) \leq w(v_{j}) + d_{C}(v_{i}, v_{j})$. As $v_{j}$ and $v_{k}$ do not dominate each other, $w(v_{j}) \leq w(v_{k}) + d_{C}(v_{j}, v_{k})$.
        Since $d_{C}(v_{i}, v_{k}) = d_{C}(v_{i}, v_{j}) + d_{C}(v_{j}, v_{k})$, we have
        \begin{align*}
            w(v_{i}) &\leq w(v_{j}) + d_{C}(v_{i}, v_{j}) \leq w(v_{k}) + d_{C}(v_{j}, v_{k}) + d_{C}(v_{i}, v_{j}) = w(v_{k}) + d_{C}(v_{i}, v_{k}),
        \end{align*}
        and
        \begin{align*}
            w(v_{k}) \leq w(v_{j}) + d_{C}(v_{j}, v_{k}) \leq w(v_{i}) + d_{C}(v_{i}, v_{j}) + d_{C}(v_{j}, v_{k}) = w(v_{i}) + d_{C}(v_{i}, v_{k}).
        \end{align*}
        Therefore, $v_{i}$ and $v_{k}$ do not dominate each other.
    \end{enumerate}
    This proves the observation.
\medskip

We are now in a position to prove the lemma.

\subparagraph{Proving the first lemma statement.}
We start with the first lemma statement. Consider a vertex $v_i$ of $C$ that is not in $S$. According to our algorithm, $v_i$ may or may not be processed in the second stage. If $v_i$ is processed in the second stage, then $v_i$ was pushed into $S$ during the first stage but is removed from $S$ in the second stage (hence $v_i$ was processed twice in the algorithm). If $v_i$ is not processed in the second stage, then $v_i$ was not in $S$ at the end of the first stage (hence $v_i$ is processed only once in the algorithm). In either case, $v_i$ must be dominated by a vertex $v_j$ that was in $S$. If $v_j$ is still in $S$ at the end of the algorithm, then the first lemma statement is proved; otherwise, we can prove inductively that $v_j$ is dominated by a vertex $v_k$ in $S$. By the above Observation, $v_i$ is dominated by $v_k$ and thus the first lemma statement follows.

\subparagraph{Proving the second lemma statement.}
We next prove the second lemma statement.
   We first prove a {\bf claim}: at any moment during the algorithm, for any two vertices $v$ and $u$ of $S$ such that $v$ is above $u$ in the stack $S$ and $u\in C_{ccw}(v)$, $v$ and $u$ do not dominate each other. We use mathematical induction to prove it, as follows.

    The claim is vacuously true in the beginning of the algorithm because $S=\emptyset$. 
    We assume that the claim holds on $S$ right before a vertex $v_i$ is processed. We show below that the claim still holds on $S$ after $v_i$ is processed. We first consider the processing of $v_i$ in the first stage of the algorithm. Let $S$ refer to the stack right before $v_i$ is processed. Let $v$ be the top element of $S$ if $S\neq \emptyset$. According to our algorithm, $v_i$ is pushed into $S$ in the following four cases.
    \begin{enumerate}
        \item $S = \emptyset$. In this case, $S=\{v_i\}$ after $v_i$ is processed. Hence, the claim trivially holds.
        \item $S \neq \emptyset$, and $v \notin C_{ccw}(v_i)$. In this case, $C_{ccw}(v_i) \cap S = \emptyset$, and thus the claim holds after $v_i$ is processed.
        \item $S \neq \emptyset$, $v \in C_{ccw}(v_i)$, and $v$ and $v_i$ do not dominate each other. Let $v_j$ be any vertex of $S$ such that $v_j\in C_{ccw}(v_i)$. To prove the lemma, it suffices to show that $v_j$ and $v_i$ do not dominate each other.

        Indeed, if $v_j=v$, then we know that $v$ and $v_i$ dot not dominate each other. Otherwise, $v_j$ is below $v$. As $v_j\in C_{ccw}(v_i)$, since $v\in C_{ccw}(v_i)$, $v_j$ is also in $C_{ccw}(v)$. Since $v$ is above $v_j$ in $S$, by the induction hypothesis $v_j$ and $v$ do not dominate each other. Since both $v_j$ and $v$ are in $C_{ccw}(v_i)$, $d_C(v_i,v_j)=d_C(v_i,v)+d_C(v,v_j)$ holds. By the above Observation, $v_i$ and $v_j$ do not dominate each other. Hence, the claim holds after $v_i$ is processed.

        \item $S \neq \emptyset$, $v \in C_{ccw}(v_i)$, and $v_i$ dominates $v$. In this case, $v_i$ is pushed into $S$ after some vertices including $v$ are popped out of $S$. First of all, since the claim holds on $S$, after vertices popped out of $S$, the claim still holds on the new $S$. Let $S$ refer to the stack right before $v_i$ is pushed in. Hence, the claim holds on $S$.

        If $S=\emptyset$, then the claim still holds after $v_i$ is pushed in since $v_i$ will be the only vertex in $S$. Otherwise, let $v'$ be the top element of $S$. If $v'\not\in C_{ccw}(v_i)$, then the claim still holds on $S$ after $v_i$ is pushed in. Otherwise, $v_i$ does not dominate $v'$ and both $v$ and $v'$ are in $C_{ccw}(v_i)$, and thus $v'$ is in $C_{ccw}(v)$. By the induction hypothesis, $v$ and $v'$ do not dominate each other. As $v_i$ dominates $v$, $v'$ cannot dominate $v_i$ since otherwise $v'$ would dominate $v$ by the above Observation. Hence, $v_i$ and $v'$ do not dominate each other. By the same argument as the above third case, $v_j$ and $v_i$ do not dominate each other for any vertex $v_j$ of $S$ with $v_j\in C_{ccw}(v_i)$. Hence, the claim holds after $v_i$ is pushed into $S$.
    \end{enumerate}

    The above proves that the claim still holds after $v_i$ is processed in the first stage of the algorithm. Now consider processing $v_i$ in the second stage. If $v_i\not\in C_{cw}(v_1)$, then the algorithm stops without changing $S$ and thus the claim still holds on $S$. Otherwise, $v_i$ is removed from $S$, after which the claim still holds on $S$. Next, the algorithm processes $v_i$ in the same way as in the first stage and thus we can use the same argument as above to prove that the claim still holds after $v_i$ is processed.
    This proves the claim.

    \medskip
    In the sequel we prove the second lemma statement by using the claim.

    Consider two vertices $v_i$ and $v_j$ in $S$ at the end of the algorithm. Notice that either $v_i\in C_{ccw}(v_j)$ or $v_j\in C_{ccw}(v_i)$. Without loss of generality, we assume that the former case holds. Our goal is to show that $v_i$ and $v_j$ do not dominate each other.

    If $v_j$ is above $v_i$ in $S$ at the end of the algorithm, then by the above claim, $v_j$ and $v_i$ do not dominate each other. Otherwise, according to our algorithm, $v_i$ must be last pushed into $S$ in the second stage while $v_j$ must be last pushed into $S$ in the first stage. Further, $v_i$ was also pushed into $S$ in the first stage before $v_j$ was processed and $v_i$ was never popped out of $S$ in the first stage. Hence, at the moment of the algorithm right before $v_j$ was processed, $v_i$ was already in $S$. At the moment of the algorithm right after $v_j$ was processed, $v_j$ was at the top of $S$ and $v_i$ was also in $S$. Since $v_i\in C_{ccw}(v_j)$, by the above claim, $v_j$ and $v_i$ do not dominate each other.
    This proves the second lemma statement.
\end{proof}

\subsection{Computing the diameter}

In the following, we use $S$ to refer to the stack after the pruning algorithm. Note that $S$ cannot be empty. The following lemma shows how $S$ can help to find a vertex-weighted diametral pair of $C$.

\begin{lemma}\label{lem:sizeS}
If $|S|=1$, then any vertex-weighted diametral pair of $C$ must contain the only vertex in $S$. Otherwise, for any vertex $v$ of $C$ that is not in $S$, $v$ cannot be in any vertex-weighted diametral pair of $C$.
\end{lemma}
\begin{proof}
Suppose $|S|=1$ and let $v$ be the only vertex in $S$. Then, by Lemma~\ref{lem:pruning}, every vertex of $C\setminus\{v\}$ is dominated by $v$. Let $(u^*,v^*)$ be a vertex-weighted diametral pair of $C$. At least one of $u^*$ and $v^*$ is not $v$. Without loss of generality, we assume $u^*\neq v$. Hence, $u^*$ is dominated by $v$. Since $(u^*,v^*)$ is a vertex-weighted diametral pair, by Lemma~\ref{lemma:1}, $v^*$ must be $v$. This proves the lemma for the case $|S|=1$.

Now assume $|S|>1$. Let $v$ be a vertex of $C$ that is not in $S$. By Lemma~\ref{lem:pruning}, $S$ has a vertex $u$ that dominates $v$.
Assume to the contrary that $v$ is in a vertex-weighted diametral pair. Then by Lemma~\ref{lemma:1}, the pair must be $(u,v)$. As $|S|>1$, $S$ has another vertex $u'$ that is not in $\{u,v\}$. Since $u$ dominates $v$, we have $w(u)>w(v)+d_C(u,v)$. Since $u$ and $u'$ do not dominate each other, we have $w(u)\leq w(u')+d_C(u,u')$. Consequently, we can derive
\begin{align*}
w(u)+d_C(u,u')+w(u') \geq w(u) + w(u) >w(u) + d_C(u,v) + w(v).
\end{align*}
But this incurs contradiction since $(u,v)$ is a vertex-weighted diametral pair.
\end{proof}

In light of Lemma~\ref{lem:sizeS}, if $|S|=1$, we compute the diameter $\Delta(G)$ as follows. Let $v$ be the only vertex in $S$. We find the vertex $u\in C\setminus\{v\}$ that maximizes the value $w(u)+d_C(u,v)+w(v)$, which can be done in $O(n)$ time with Lemma~\ref{lem:preprocess}. By Lemma~\ref{lem:sizeS}, $(u,v)$ is a vertex-weighted diametral pair and $\Delta(G)=w(u)+d_C(u,v)+w(v)$.

If $|S|>1$, by Lemma~\ref{lem:sizeS}, $\Delta(G)=\max_{u,v\in S}(w(u)+d_C(u,v)+w(v))$.
The following lemma finds a vertex-weighted diametral pair and thus computes $\Delta(G)$ in linear time.

\begin{lemma}\label{lem:pair}
A pair $(u,v)$ of vertices in $S$ that maximizes the value $w(u)+d_C(u,v)+w(v)$ can be found in $O(n)$ time.
\end{lemma}
\begin{proof}
Consider a vertex $u \in S$. Let $v_i$ and $v_j$ be two vertices in $S\cap C_{cw}(u)$ such that $d_C(u,v_i)< d_C(u,v_j)$ (e.g., see Fig.~\ref{fig:pair}). Note that $d_{C}(u, v_{j})=d_{C}(u, v_{i})+d_{C}(v_i, v_{j})$.
We claim that $w(u)+d_C(u,v_i)+w(v_i)\leq w(u)+d_C(u,v_j)+w(v_j)$. Indeed, by Lemma~\ref{lem:pruning}(2), $v_i$ and $v_j$ do not dominate each other. Hence,
$w(u) + d_{C}(u, v_{i}) + w(v_{i}) \leq w(u)+ d_{C}(u, v_{i}) + w(v_{j}) + d_{C}(v_{i}, v_{j}) = w(u) + d_{C}(u, v_{j}) + w(v_{j})$. The claim follows.
The claim implies that if we consider the vertices $v$ of $S\cap C_{cw}(u)$ from $u$ along $C$ in clockwise order, then the value $w(u)+d_C(u,v)+w(v)$ is monotonically increasing. Similarly, if we consider the vertices $v$ of $S\cap C_{ccw}(u)$ from $u$ along $C$ in counterclockwise order, then the value $w(u)+d_C(u,v)+w(v)$ is monotonically increasing. Let $u^{cw}$ (resp., $u^{ccw}$) refer to the farthest vertex from $u$ in $C_{cw}(u)$ (resp., $C_{ccw}(u)$); e.g., see Fig.~\ref{fig:pair}.
Based on the above discussion, it holds that $\max_{v\in S\setminus\{u\}}(w(u)+d_C(u,v)+w(v))=\max\{w(u)+d_C(u,u^{cw})+w(u^{cw}),w(u)+d_C(u,u^{ccw})+w(u^{ccw})\}$.

\begin{figure}[t]
    \centering
    \includegraphics[width=1.8in]{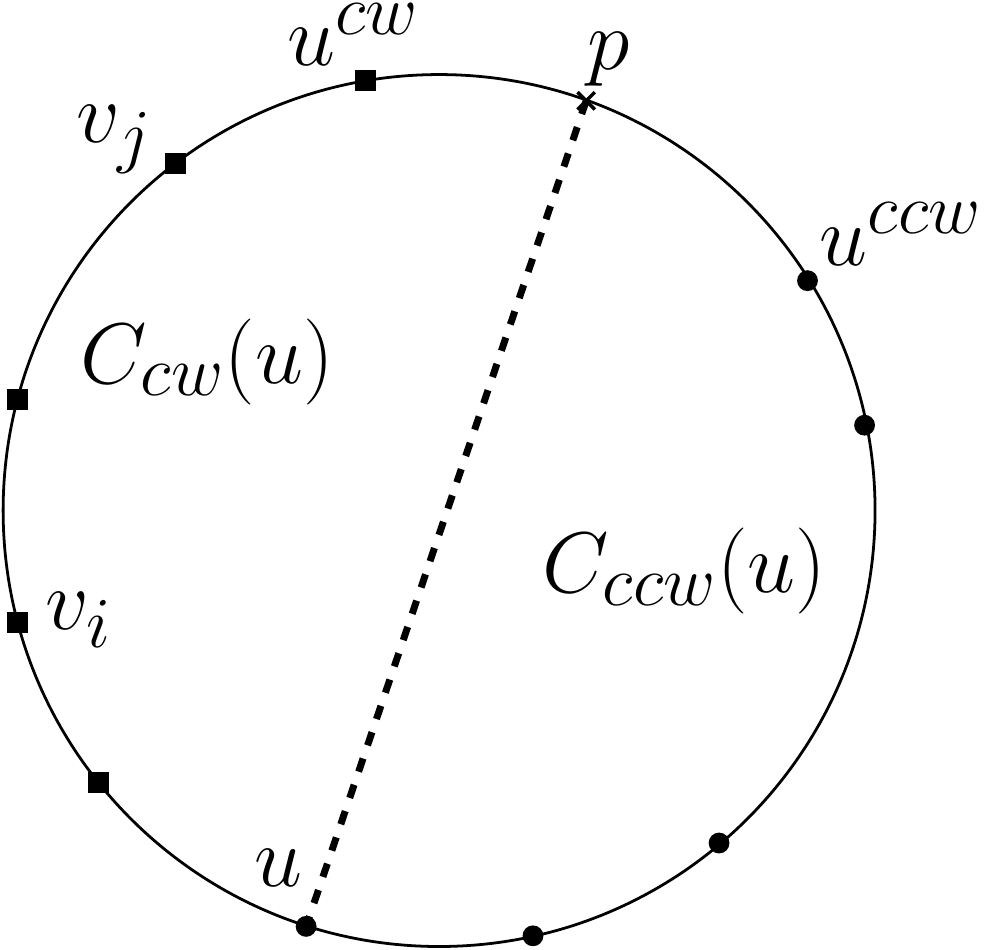}
    \caption{Illustrating the proof of Lemma~\ref{lem:pair}: we assume that $p$ is a point on $C$ that together with $u$ partitions $C$ into two half-cycles of equal length.}
    \label{fig:pair}
\end{figure}

To find  $u^{cw}$ and  $u^{ccw}$, notice that if we traverse the vertices $u$ of $S$ along $C$ in clockwise order, then both $u^{cw}$ and  $u^{ccw}$ are also ordered along $C$ in clockwise order. Hence, $u^{cw}$ and  $u^{ccw}$ for all vertices $u$ of $S$ can be found in total $O(n)$ time by traversing the vertices of $S$ along $C$. Consequently, we have $\max_{u,v\in S}((w(u)+d_C(u,v)+w(v))=\max_{u\in S}\max\{w(u)+d_C(u,u^{cw})+w(u^{cw}),w(u)+d_C(u,u^{ccw})+w(u^{ccw})\}$, which can be computed in $O(n)$ time.
\end{proof}

The proof of the following theorem summarizes our algorithm.

\begin{theorem}
    \label{theorem:1}
    The diameter (along with a diametral pair) of a unicycle graph can be computed in linear time.
\end{theorem}
\begin{proof}
Recall that there are two cases for a diametral pair $(u^*,v^*)$ of $G$: (1) both $u^*$ and $v^*$ are in the same tree of $\calT(G)$; (2) $u^*$ and $v^*$ are in two different trees of $\calT(G)$. 

\begin{itemize}
    \item 
    The first case can be handled by computing the diameter and the corresponding diametral pair of each tree of $\calT(G)$, which can be done in total $O(n)$ time. The longest diameter of these trees is kept as a candidate diameter of $G$ and the corresponding diametral pair is kept as a candidate diametral pair. 
    
    \item
    For the second case, we first compute the weights $w(v_i)$ for all vertices $v_i$ of the cycle $C$; for each $v_i$, we also store its farthest vertex $f(v_i)$ in $T(v_i)$. Then, we perform the preprocessing of Lemma~\ref{lem:preprocess}. Next, we run the pruning algorithm to obtain $S$. If $|S|=1$, we find a vertex-weighted diametral pair of $C$ as described above; otherwise, we use Lemma~\ref{lem:pair} to find such a pair. In either case, let $(u,v)$ denote the pair. Then, a candidate diameter is $w(u)+d_C(u,v)+w(v)$. In addition, $(f(u),f(v))$ is a candidate diametral pair of $G$.
\end{itemize}

We compare the candidate diameters obtained from the above two cases and return the larger one as the diameter of $G$; the corresponding diametral pair is a diametral pair of $G$. The running time of the overall algorithm is $O(n)$.
\end{proof}

\section{The Diameter-Optimally Augmenting Trees (DOAT)}
\label{sec:DOAPForATree}

In this section, we solve the DOAT problem in $O(n^2\log n)$ time and $O(n)$ space. Our algorithm for computing the diameter of a unicycle graph will be used as a subroutine.

\subsection{Observations}
We follow the same notation as in Section~\ref{sec:ComputingTheDiameterOfAUnicylicGraph} such as $\pi_{G'}(s,t)$, $d_{G'}(s,t)$, $\Delta(G)$.

Let $T$ be a tree of $n$ vertices such that each edge of $T$ has a positive length. For any two vertices $u$ and $v$ of $T$, we use $e(u,v)$ to refer to the shortcut connecting $u$ and $v$; note that even if $T$ already has an edge connecting them, we can always assume that there is an alternative shortcut (or we could also consider the shortcut as the edge itself with the same length).
Let $|e(u,v)|$ denote the length of $e(u,v)$. Again, there is an oracle that can return the value $|e(u,v)|$ in $O(1)$ time for any shortcut $e(u,v)$.
Denote by $T+e(u,v)$ the graph after adding $e(u,v)$ to $T$. The goal of the DOAT problem is to find a shortcut $e(u,v)$ so that the diameter of the new graph $\Delta(T+e(u,v))$ is minimized. Let $\Delta^*(T)$ be the diameter of an optimal solution. In the following we assume that $\Delta^*(T)<\Delta(T)$, since otherwise any shortcut would be sufficient.

For any shortcut $e(u,v)$, $T$ has a unique path $\pi_T(u,v)$ between $u$ and $v$. We make an assumption that $|e(u,v)|< d_T(u,v)$ since otherwise $e(u,v)$ can never be used (indeed, whenever $e(u,v)$ was used in a shortest path, we could always replace it with $\pi_T(u,v)$ to get a shorter path). This assumption is only for the argument of the correctness of our algorithm; the algorithm itself still uses the true value of $|e(u,v)|$ (this does not affect the correctness, because if $|e(u,v)|\geq d_T(u,v)$, then $e(u,v)$ cannot be an optimal shortcut). For the reference purpose, we refer to this assumption as the {\em shortcut length assumption}.

At the outset, we compute a diametral path $P$ of $T$ in $O(n)$ time. Let $v_1,v_2,\ldots,v_m$ be the vertices of $P$ ordered along it. Removing the edges of $P$ from $T$ results in $m$ connected components of $T$, each of which is a tree containing a vertex of $P$; we let $T(v_i)$ denote the tree containing $v_i$. For each $v_i$, we define a weight $w(v_i)$ as the distance from $v_i$ to its farthest vertex in $T(v_i)$.
Let $\calT=\{T(v_i)\ | \ 1\leq i\leq m\}$.

For any pair $(i,j)$ of indices with $1\leq i<j\leq m$, we define a {\em critical pair} of vertices $(x,y)$ with $x\in T(v_i)$ and $y\in T(v_j)$ such that they minimize the value $d_{T(v_i)}(v_i,x')+|e(x',y')|+d_{T(v_j)}(y',v_j)$ among all vertex pairs $(x',y')$ with $x'\in T(v_i)$ and $y'\in T(v_j)$.

The following lemma will be used on several occasions later on.
\begin{lemma}\label{lem:10}
For any vertex $v$ in any tree $T(v_{k}) \in \calT$, it holds that $d_{T(v_k)}(v,v_k) \leq \\ \min\{d_T(v_1,v_k),d_T(v_k,v_m)\}$.
Also, $d_T(v_1,v_k)=d_P(v_1,v_k)$ and $d_T(v_k,v_m)=d_P(v_k,v_m)$.
\end{lemma}
\begin{proof}
Assume to the contrary that $d_{T(v_k)}(v,v_k)> \min\{d_T(v_1,v_k),d_T(v_k,v_m)\}$. Without loss of generality, we assume  $d_{T(v_k)}(v,v_k)>d_T(v_1,v_k)$. Then, $d_T(v,v_m)=d_{T(v_k)}(v,v_k)+d_T(v_k,v_m)>d_T(v_1,v_k)+d_T(v_k,v_m)=d_T(v_1,v_m)=\Delta(T)$, a contradiction.

The second part of the lemma holds because $P$ is a path of $T$.
\end{proof}

The following lemma demonstrates why critical pairs are ``critical''.

\begin{lemma}
    \label{lemma:6}
   Suppose $e(u^*, v^*)$ is an optimal shortcut with $u^*\in T(v_i)$ and $v^*\in T(v_j)$. Then, $i\neq j$ and any critical pair of $(i, j)$ also defines an optimal shortcut.
\end{lemma}
\begin{proof}
Because $P$ is a diametral path of $T$, if $i = j$, then $P$ is still the shortest path from $v_1$ to $v_m$ in the new graph $T+e(u^*,v^*)$, and thus
we have $\Delta^*=d_{T + e(u^*, v^*)}(v_{1}, v_{m})=d_{T}(v_{1}, v_{m})=\Delta(T)$. But this contradicts with our assumption $\Delta^*<\Delta(T)$. Hence, $i\neq j$.

Let $(x,y)$ be a critical pair of $(i,j)$ with $x\in T(v_i)$ and $y\in T(v_j)$.
Define $\phi(v_i, v_j) = d_{T(v_i)}(v_{i}, x) + |e(x, y)| + d_{T(v_j)}(v_{j}, y)$ and $\phi^*(v_i, v_j) = d_{T(v_i)}(v_{i}, u^*) + |e(u^*, v^*)| + d_{T(v_j)}(v_{j}, v^*)$.
By the definition of critical pairs, we have $\phi(v_i, v_j) \leq \phi^*(v_i, v_j)$. Also, due to the shortcut length assumption, it holds that $\phi(v_i,v_j)\leq |e(v_i,v_j)|< d_{T}(v_{i}, v_{j})$.

In the following, we prove that $d_{T + e(x, y)}(u, v) \leq \Delta(T + e(u^*, v^*))$ for any two vertices $u, v \in T$. This will prove the lemma. To simplify the notation, let $T_{xy}=T+e(x,y)$ and $T^*=T + e(u^*, y^*)$.

First of all, if $d_{T^*}(u, v) = d_{T}(u, v)$, then $d_{T_{xy}}(u, v) \leq d_T(u,v) = d_{T^*}(u, v)\leq \Delta(T^*)$. Below we assume $d_{T^*}(u, v) \neq d_{T}(u, v)$. Thus, $d_{T^*}(u, v) < d_{T}(u, v)$ and the shortest path $\pi_{T^*}(u, v)$ must contain the shortcut $e(u^*,v^*)$.

Let $T(v_k)$ and $T(v_h)$ be the trees of $\calT$ that contain $u$ and $v$, respectively, with $k\leq h$.
Based on relationship of the indices $k$, $h$, $i$, and $j$, there are several cases.

    \begin{enumerate}
        \item $\{k, h\}\cap \{i, j\}=\emptyset$; e.g., see Fig.~\ref{fig:case10}. In this case, we have
        \begin{align*}
            d_{T_{xy}}(u,v) &\leq d_{T}(u, v_{i}) + \phi(v_i, v_j) + d_{T}(v_{j}, v) \\
            &\leq d_{T}(u, v_{i}) + \phi^*(v_i, v_j) + d_{T}(v_{j}, v) \\
            & = d_{T^*}(u, v) \leq \Delta(T^*).
        \end{align*}
        Note that $d_{T}(u, v_{i}) + \phi^*(v_i, v_j) + d_{T}(v_{j}, v)=d_{T^*}(u, v)$ because $\pi_{T^*}(u, v)$ contains $e(u^*,v^*)$ and $k\leq h$.

\begin{figure}[t]
\begin{minipage}[t]{0.55\textwidth}
\begin{center}
\includegraphics[height=1.4in]{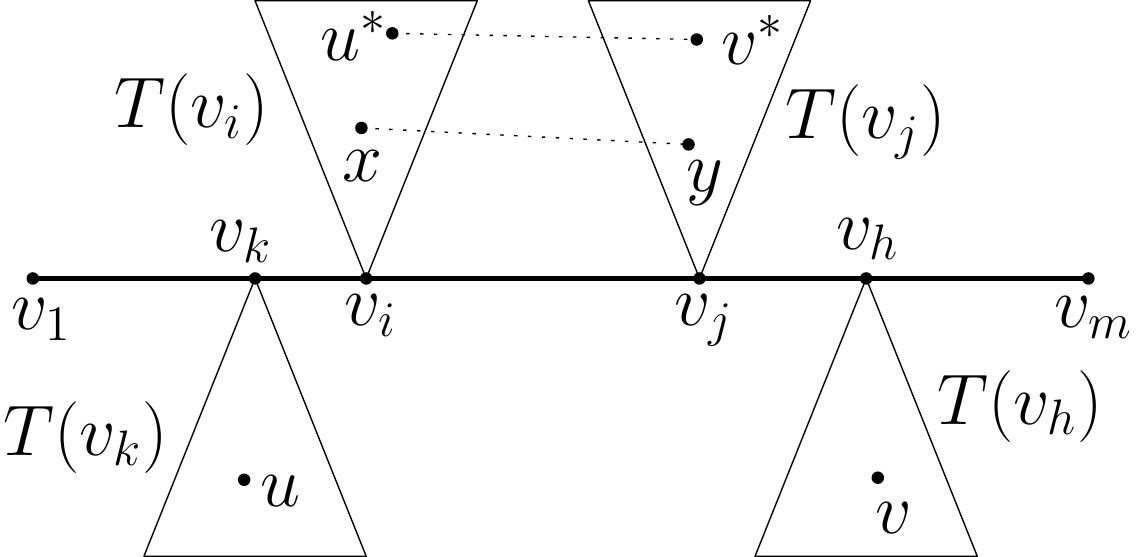}
\caption{\footnotesize Illustrating the case $\{k, h\}\cap \{i, j\}=\emptyset$.}
\label{fig:case10}
\end{center}
\end{minipage}
\hspace{0.05in}
\begin{minipage}[t]{0.44\textwidth}
\begin{center}
\includegraphics[height=0.9in]{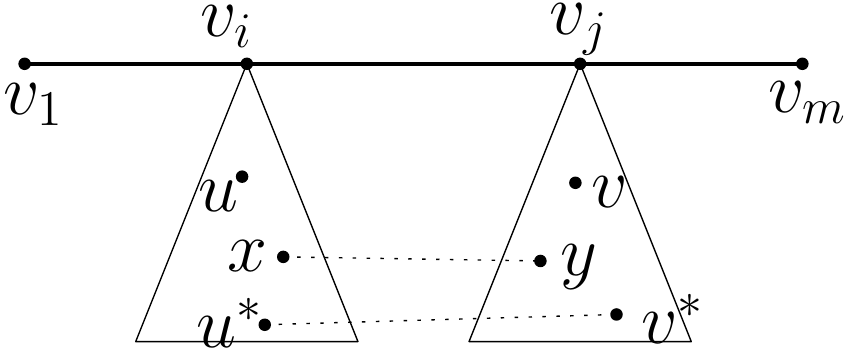}
\caption{\footnotesize Illustrating the case $k = i$ and $h = j$.}
\label{fig:case20}
\end{center}
\end{minipage}
\vspace{-0.15in}
\end{figure}

        \item $k = i$ and $h = j$; e.g., see Fig.~\ref{fig:case20}. In this case, we have
        \begin{align*}
            d_{T_{xy}}(u, v) &\leq d_{T(v_{i})}(u, x) + |e(x, y)| + d_{T(v_{j})}(y, v) \\
            &\leq d_{T(v_i)}(u, v_{i}) + d_{T(v_{i})}(v_{i}, x) + |e(x, y)| + d_{T(v_{j})}(y,v_{j}) + d_{T(v_j)}(v_{j}, v).
        \end{align*}
        By Lemma~\ref{lem:10}, $d_{T(v_i)}(u, v_{i})\leq d_T(v_1,v_i)$ and $d_{T(v_j)}(v_{j}, v)\leq d_T(v_j,v_m)$. Hence, we can obtain
        \begin{align*}
            d_{T_{xy}}(u, v) & \leq d_{T}(v_{1}, v_{i}) +  d_{T(v_{i})}(v_{i}, x) + |e(x, y)| + d_{T(v_{j})}(y,v_{j}) +d_T(v_j,v_m) \\
            &= d_{T}(v_{1}, v_{i}) + \phi(v_i, v_j) + d_{T}(v_{j}, v_m) \\
            &\leq d_{T}(v_{1}, v_{i}) + \phi^*(v_i, v_j) + d_{T}(v_{j}, v_m) \\
            & = d_{T^*}(v_{1}, v_m) \leq \Delta(T^*).
        \end{align*}

\begin{figure}[t]
\begin{minipage}[t]{0.49\textwidth}
\begin{center}
\includegraphics[height=1.4in]{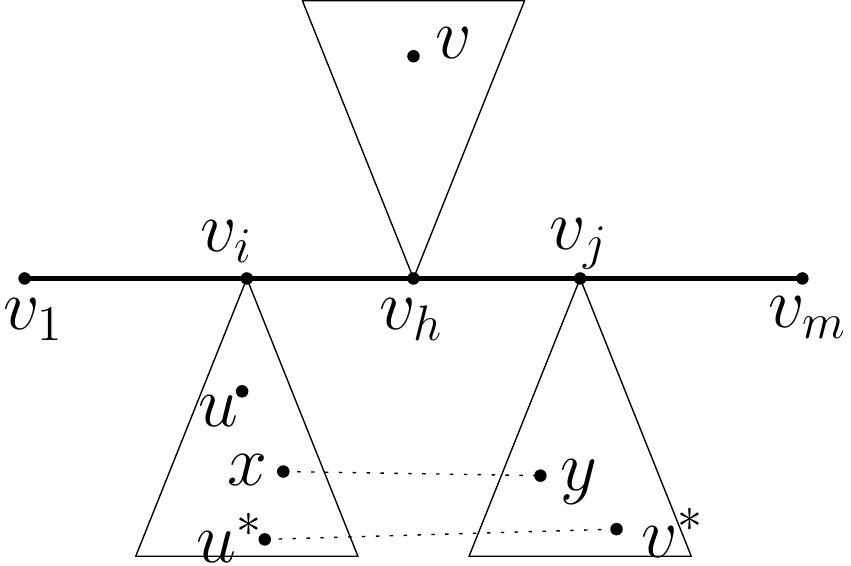}
\caption{\footnotesize Illustrating the case $k = i$, $h > i$, and $h\neq j$.}
\label{fig:case30}
\end{center}
\end{minipage}
\hspace{0.05in}
\begin{minipage}[t]{0.49\textwidth}
\begin{center}
\includegraphics[height=1.0in]{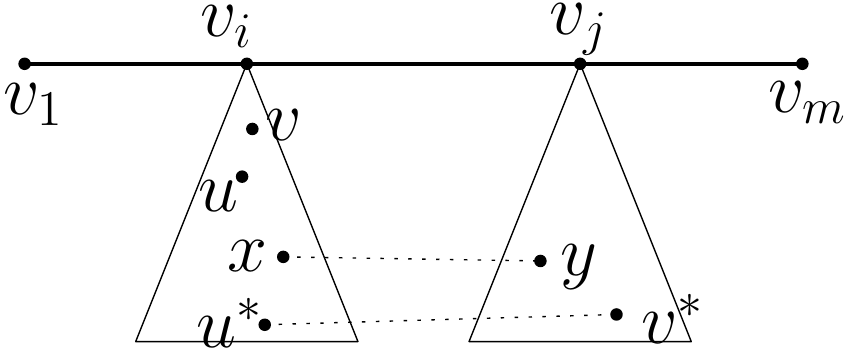}
\caption{\footnotesize Illustrating the case $k = h = i$.}
\label{fig:case40}
\end{center}
\end{minipage}
\vspace{-0.15in}
\end{figure}

        \item $k = i$, $h > i$, and $h\neq j$; e.g., see Fig.~\ref{fig:case30}. In this case, we have
        \begin{align*}
            d_{T_{xy}}(u, v) &\leq d_{T(v_{i})}(u, x) + |e(x, y)| + d_{T(v_{j})}(y,v_{j}) + d_{T}(v_{j}, v) \\
            &\leq d_{T(v_i)}(u,v_i)+ d_{T(v_{i})}(v_i, x) + |e(x, y)| + d_{T(v_{j})}(y,v_{j}) + d_{T}(v_{j}, v) \\
            &\leq d_{T}(v_{1}, v_{i}) + d_{T(v_{i})}(v_{i}, x) + |e(x, y)| + d_{T(v_{j})}(y,v_{j}) + d_{T}(v_{j}, v) \\
            &= d_{T}(v_{1}, v_{i}) + \phi(v_i, v_j) + d_{T}(v_{j}, v) \\
            &\leq d_{T}(v_{1}, v_{i}) + \phi^*(v_i, v_j) + d_{T}(v_{j}, v).
        \end{align*}
        On the other hand, since $h>i$, we also have
        \begin{align*}
            d_{T_{xy}}(u, v) &\leq d_{T(v_{i})}(u, v_i) + d_{T}(v_{i}, v)        \leq d_{T}(v_{1}, v_{i}) + d_{T}(v_{i}, v).
        \end{align*}
        Notice that $d_{T^*}(v_1,v)=\min\{d_{T}(v_{1}, v_{i}) + d_{T}(v_{i}, v), d_{T}(v_{1}, v_{i}) + \phi^*(v_i, v_j) + d_{T}(v_{j}, v)\}$. Combining the above two inequalities, we derive $d_{T_{xy}}(u,v)\leq d_{T^*}(v_1,v)\leq \Delta(T^*)$.

        \item $k < j$, $k\neq i$, and $h = j$. This is case is symmetric to the above case and we can prove $d_{T_{xy}}(u,v)\leq \Delta(T^*)$ by a similar argument.

        \item $k = h = i$; e.g., see Fig.~\ref{fig:case40}. We claim that this case cannot happen. Indeed, recall that $\pi_{T^*}(u,v)$ contains $e(u^*,v^*)$. Since both $u$ and $v$ are in $T(v_i)$, $\pi_{T^*}(u,v)$ must also contain $P[v_i,v_j]$, where $P[v_i,v_j]$ is the subpath of $P$ between $v_i$ and $v_j$.
        On the other hand, $\Delta(T)=d_T(v_1,v_m)$, which is equal to the length of $P$. Since $\Delta(T)>\Delta^*=\Delta(T^*)$, $\pi_{T^*}(v_1,v_m)$ must contain the shortcut $e(u^*,v^*)$ but cannot contain $P[v_i,v_j]$. This implies that $P[v_i,v_j]$ is not a shortest path between $v_i$ and $v_j$ in $T^*$. However, since $\pi_{T^*}(u,v)$, which is a shortest path between $u$ and $v$ in $T^*$, contains $P[v_i,v_j]$, $P[v_i,v_j]$ must be a shortest path between $v_i$ and $v_j$ in $T^*$. We thus obtain contradiction.

        \item $k = h = j$. This case is symmetric to the above case. By a similar argument, we can show that it cannot happen.

\begin{figure}[t]
\begin{minipage}[t]{\textwidth}
\begin{center}
\includegraphics[height=1.4in]{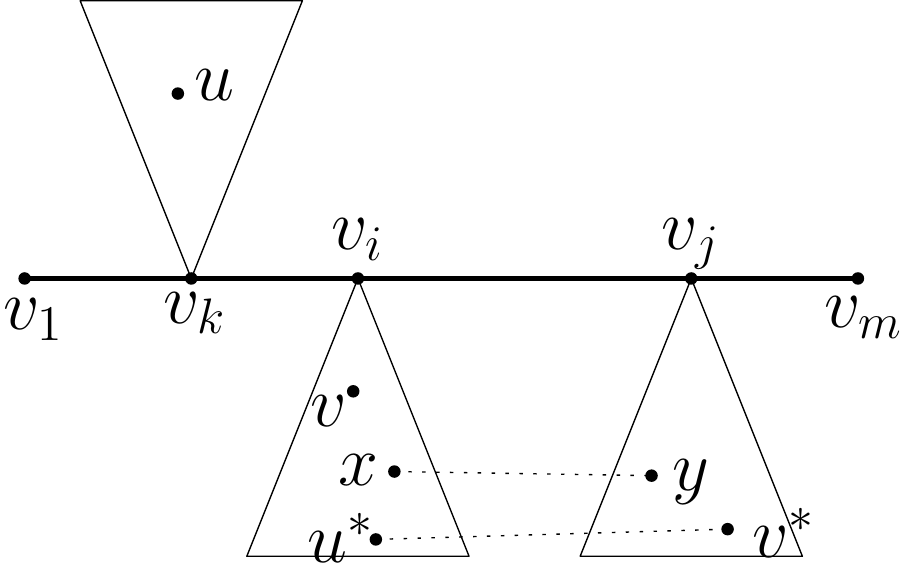}
\caption{\footnotesize Illustrating the case $k < i$ and $h = i$.}
\label{fig:case50}
\end{center}
\end{minipage}
\vspace{-0.15in}
\end{figure}

        \item $k < i$ and $h = i$; e.g., see Fig.~\ref{fig:case50}. We claim that this case cannot happen. Indeed, since $\pi_{T^*}(u,v)$ contains $e(u^*,v^*)$ and $k<i$, $\pi_{T^*}(u,v)$ must also contain the subpath of $P$ between $v_i$ and $v_j$. By the same analysis as the above fifth case, we can obtain contradiction.

        \item $k = j$ and $h > j$. This case is symmetric to the above case. By a similar argument, we can show that it cannot happen.
    \end{enumerate}


    In summary, $d_{T_{xy}}(u,v)\leq \Delta(T^*)$ holds in all possible cases. The lemma thus follows.
\end{proof}

\subsection{Reducing DOAT to finding a shortcut for $\boldsymbol{P}$}

In light of Lemma~\ref{lemma:6}, we reduce our DOAT problem on $T$ to finding a shortcut for the vertex-weighted path $P$ as follows.

For an index pair $(i,j)$ with $1\leq i<j\leq m$, we define a shortcut $\e(v_i,v_j)$ connecting $v_i$ and $v_j$ with length $|\e(v_i,v_j)|= d_{T(v_{i})}(v_{i}, x) + |e(x, y)| + d_{T(v_{j})}(v_{j}, y)$, where $(x,y)$ is a critical pair of $(i,j)$. The diameter $\Delta(P+\e(v_i,v_j))$ is defined as $\max_{1\leq k<h\leq m}\{w(v_k)+d_{P+\e(v_i,v_j)}(v_k,v_h)+w(v_h)\}$. The diameter-optimally augmenting path (DOAP) problem on $P$ is to find a shortcut $\e(v_i,v_j)$ so that the diameter $\Delta(P+\e(v_i,v_j))$ is minimized; we use $\Delta^*(P)$ to denote the minimized diameter.

With the help of Lemma~\ref{lemma:6}, the following lemma shows that the DOAT problem on $T$ can be reduced to the DOAP problem on $P$. Similar problem reductions were also used in~\cite{ref:GrobeFa15,ref:BiloAl18}. 

\begin{lemma}
    \label{lemma:5}
    \begin{enumerate}
        \item For any index pair $(i,j)$ with $1\leq i<j\leq m$, it holds that $\Delta(P+\e(v_i,v_j))=\Delta(T+e(x,y))$, where $(x,y)$ is a critical pair of $(i,j)$.
        \item If $\e(v_i,v_j)$ is an optimal shortcut for the DOAP problem on $P$, then $e(x,y)$ is an optimal shortcut for the DOAT problem on $T$, where $(x,y)$ is a critical pair of $(i,j)$.
        \item $\Delta^*(T)=\Delta^*(P)$.
    \end{enumerate}
\end{lemma}
\begin{proof}
Before proving the first lemma statement, we first use it to prove the second and third lemma statements.

\subparagraph{The second lemma statement.}
Suppose $\e(v_i,v_j)$ is an optimal shortcut for the DOAP problem on $P$. By the first lemma statement, we have $\Delta(P+\e(v_i,v_j))=\Delta(T+e(x,y))$. Assume to the contrary that $e(x,y)$ is not an optimal shortcut for the DOAT problem for $T$, and let $e(u^*,v^*)$ instead be an optimal shortcut, with $u^*\in T(v_k)$ and $v^*\in T(v_h)$. Then, $\Delta(T+e(x,y))>\Delta(T+e(u^*,v^*))$.
By Lemma~\ref{lemma:6}, $k\neq h$ and the critical pair $(x',y')$ of $(k,h)$ also defines an optimal shortcut for $T$. Hence, $\Delta(T+e(u^*,v^*))=\Delta(T+e(x',y'))$. By the first lemma statement, $\Delta(P+\e(v_k,v_h))=\Delta(T+e(x',y'))$. Combining all above, we obtain
\begin{align*}
    \Delta(P+\e(v_k,v_h))=\Delta(T+e(x',y')) = \Delta(T+e(u^*,v^*)) < \Delta(T+e(x,y)) = \Delta(P+\e(v_i,v_j)).
\end{align*}
But this incurs contradiction since $\e(v_i,v_j)$ is an optimal shortcut for the DOAP problem on $P$. Hence, the second lemma statement is proved.

\subparagraph{The third lemma statement.}
The third lemma statement follows immediately from the first two lemma statements. Indeed, suppose $\e(v_i,v_j)$ is an optimal solution for the DOAP problem on $P$. Then, by the second lemma statement, $e(x,y)$ is an optimal shortcut for the DOAT problem on $T$, where $(x,y)$ is a critical pair of $(i,j)$. Hence, we have $\Delta(P^*)=\Delta(P+\e(v_i,v_j))$ and $\Delta(T^*)=\Delta(T+e(x,y))$. By the first lemma statement, $\Delta(P+\e(v_i,v_j))=\Delta(T+e(x,y))$. Hence, $\Delta^*(T)=\Delta^*(P)$.

\subparagraph{The first lemma statement.}
We now prove the first lemma statement. To simplify the notation, let $P'=P+\e(v_i,v_j)$ and $T'=T+e(x,y)$. Our goal is to prove $\Delta(P')=\Delta(T')$.

First of all, because $(x,y)$ is a critical pair of $(i,j)$, $d_{P'}(v_k,v_h)=d_{T'}(v_k,v_h)$ holds for any two vertices $v_k$ and $v_h$ of $P$ with $1\leq k<h\leq m$.

We first prove that $\Delta(P') \leq \Delta(T')$. Let $v_k$ and $v_h$ be any two vertices of $P$ with $1 \leq k < h \leq m$. It suffices to show that $w(v_{k}) + d_{P'}(v_{k}, v_{h}) + w(v_{h})\leq \Delta(T')$.

Let $u$ be the vertex of $T(v_{k})$ farthest from $v_k$, i.e.,
$w(v_{k}) = d_{T(v_{k})}(v_k, u)$. Similarly, let $v$ be the vertex of $T(v_{h})$ farthest from $v_h$. Recall that $i<j$ and $k<h$. Depending on the values of $i,j,k,h$, there are several cases.

\begin{figure}[t]
\begin{minipage}[t]{0.49\textwidth}
\begin{center}
\includegraphics[height=0.9in]{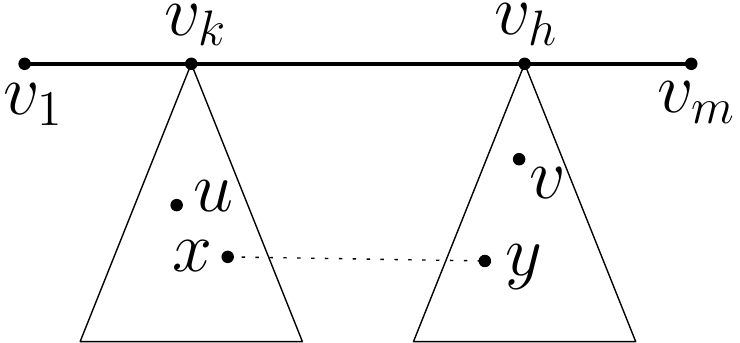}
\caption{\footnotesize Illustrating the case $i=k$ and $j=h$.}
\label{fig:reduction10}
\end{center}
\end{minipage}
\hspace{0.05in}
\begin{minipage}[t]{0.49\textwidth}
\begin{center}
\includegraphics[height=1.4in]{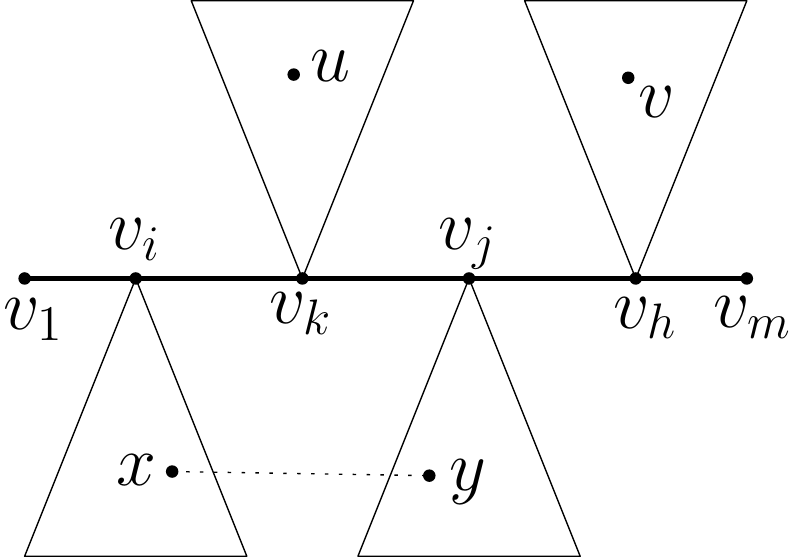}
\caption{\footnotesize Illustrating the case $\{i,j\}\cap \{k,h\}=\emptyset$.}
\label{fig:reduction20}
\end{center}
\end{minipage}
\vspace{-0.15in}
\end{figure}

\begin{enumerate}
\item $i=k$ and $j=h$; e.g., see Fig.~\ref{fig:reduction10}. In this case, $d_{T'}(v_1,v_k)=d_{T}(v_1,v_k)$ and $d_{T'}(v_h,v_m)=d_{T}(v_h,v_m)$. By Lemma~\ref{lem:10}, $d_{T(v_k)}(u,v_k)\leq d_T(v_1,v_k)$ and $d_{T(v_h)}(v,v_h)\leq d_T(v_h,v_m)$.
Hence, we can derive
\begin{align*}
    w(v_{k}) + d_{P'}(v_{k}, v_{h}) + w(v_{h}) & = d_{T(v_{k})}(u, v_k) + d_{T'}(v_k, v_h) + d_{T(v_{h})}(v_h, v) \\
    & \leq d_T(v_1,v_k) + d_{T'}(v_k, v_h) + d_T(v_h,v_m)\\
    & = d_{T'}(v_1,v_k) + d_{T'}(v_k, v_h) + d_{T'}(v_h,v_m)\\
    & = d_{T'}(v_1,v_m)\leq \Delta(T').
\end{align*}

\item $\{i,j\}\cap \{k,h\}=\emptyset$; e.g., see Fig.~\ref{fig:reduction20}. In this case, the shortest path $\pi_{T'}(u,v)$ between $u$ and $v$ in $T'$ must contain $\pi_{T(v_k)}(u,v_k)$ and $\pi_{T(v_h)}(v_h,v)$. Hence, we have
\begin{align*}
    w(v_{k}) + d_{P'}(v_{k}, v_{h}) + w(v_{h}) & = d_{T(v_{k})}(u, v_k) + d_{T'}(v_k, v_h) + d_{T(v_{h})}(v_h, v) \\
    &= d_{T'}(u,v)\leq \Delta(T').
\end{align*}

\begin{figure}[t]
\begin{minipage}[t]{0.49\textwidth}
\begin{center}
\includegraphics[height=1.4in]{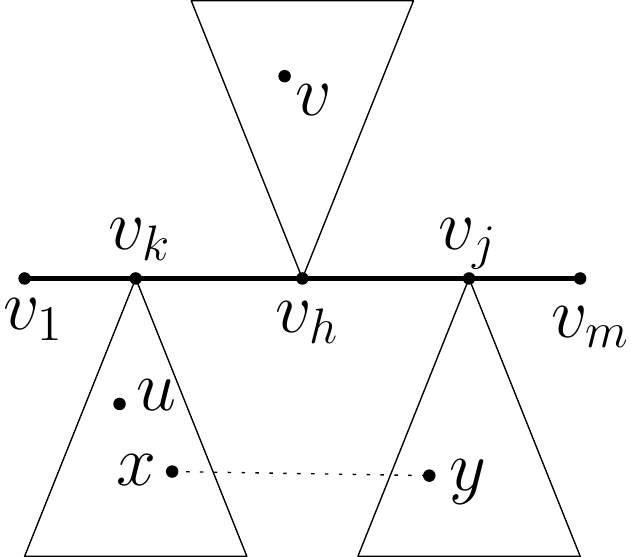}
\caption{\footnotesize Illustrating the case $i=k$ and $j\neq h$.}
\label{fig:reduction30}
\end{center}
\end{minipage}
\hspace{0.05in}
\begin{minipage}[t]{0.49\textwidth}
\begin{center}
\includegraphics[height=1.4in]{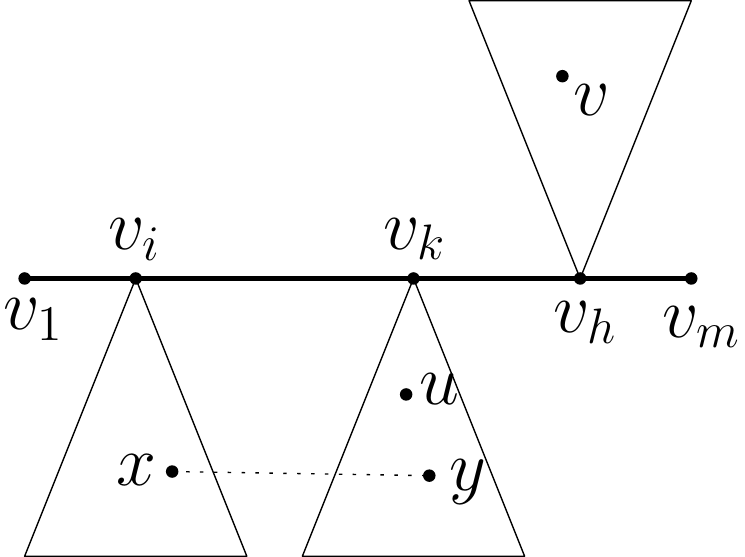}
\caption{\footnotesize Illustrating the case $j=k$.}
\label{fig:reduction40}
\end{center}
\end{minipage}
\vspace{-0.15in}
\end{figure}

\item $i=k$ and $j\neq h$; e.g., see Fig.~\ref{fig:reduction30}. In this case, the shortest path $\pi_{T'}(v_1,v)$ between $v_1$ and $v$ in $T'$ must contain $\pi_{T}(v_1,v_k)$ and $\pi_{T(v_h)}(v_h,v)$. Hence, we have
\begin{align*}
    w(v_{k}) + d_{P'}(v_{k}, v_{h}) + w(v_{h}) & = d_{T(v_{k})}(u, v_k) + d_{T'}(v_k, v_h) + d_{T(v_{h})}(v_h, v) \\
    & \leq d_T(v_1,v_k) + d_{T'}(v_k, v_h) + d_{T(v_{h})}(v_h, v)\\
    & = d_{T'}(v_1,v_k) + d_{T'}(v_k, v_h) + d_{T(v_{h})}(v_h, v)\\
    & = d_{T'}(v_1,v)\leq \Delta(T').
\end{align*}

\item $i\neq k$ and $j=h$. This case is symmetric to the above case. By a similar argument, we can show that $w(v_{k}) + d_{P'}(v_{k}, v_{h}) + w(v_{h})\leq \Delta(T')$.

\item $j=k$; e.g., see Fig.~\ref{fig:reduction40}. In this case, we claim that $\pi_{T(v_k)}(u,v_k)$ is $\pi_{T'}(u,v_k)$. Indeed, assume to the contrary that this is not true. Then, $\pi_{T'}(u,v_k)$ must contain the shortcut $e(x,y)$. This further implies that $\pi_{T'}(u,v_k)$ must contain $\pi_T(v_i,v_k)$, which is the subpath of $P$ between $v_i$ and $v_k$ (=$v_j$). Since $(x,y)$ is a critical pair of $(i,j)$, $\pi_T(v_i,v_k)$ is no shorter than the path $\pi'=\pi_{T(v_i)}(v_i,x)\cup e(x,y)\cup \pi_{T(v_k)}(y,v_k)$. Therefore, if we replace $\pi_T(v_i,v_k)$ by $\pi'$ in $\pi_{T'}(u,v_k)$, we can obtain another shortest path $\pi'_{T'}(u,v_k)$ from $u$ to $v_k$ in $T'$. However, $\pi'_{T'}(u,v_k)$ cannot be a shortest path as it contains $e(x,y)$ twice. This incurs contradiction and thus the claim follows.

Due to the claim, we have $d_{T'}(u,v)=d_{T(v_k)}(u,v_k)+d_{T'}(v_k,v_h)+d_{T(v_h)}(v_h,v)$. 
We can now obtain
\begin{align*}
    w(v_{k}) + d_{P'}(v_{k}, v_{h}) + w(v_{h}) & = d_{T(v_{k})}(u, v_k) + d_{T'}(v_k, v_h) + d_{T(v_{h})}(v_h, v) \\
    & = d_{T'}(u,v)\leq \Delta(T').
\end{align*}

\item $i=h$. This case is symmetric to the above case. By a similar argument, we can show that $w(v_{k}) + d_{P'}(v_{k}, v_{h}) + w(v_{h})\leq \Delta(T')$.
\end{enumerate}

This proves that $\Delta(P') \leq \Delta(T')$.

\medskip

Next we prove that $\Delta(T') \leq \Delta(P')$. Let $u$ and $v$ be any two vertices of $T$ with $u \in T(v_{k})$ and $v \in T(v_{h})$, with $1 \leq k \leq h \leq m$. It suffices to show that $d_{T'}(u,v)\leq \Delta(P')$.

If $k \neq h$, then we can deduce
    \begin{align*}
        d_{T'}(u, v) &\leq d_{T(v_{k})}(u,v_k) + d_{T'}(v_{k}, v_{h}) + d_{T(v_{h})}(v_{h}, v)\\
        &\leq w(v_{k}) + d_{T'}(v_{k}, v_{h}) + w(v_{h})\\
        &= w(v_{k}) + d_{P'}(v_{k}, v_{h}) + w(v_{h}) \leq \Delta(P').
    \end{align*}

\begin{figure}[t]
\begin{minipage}[t]{\textwidth}
\begin{center}
\includegraphics[height=1.4in]{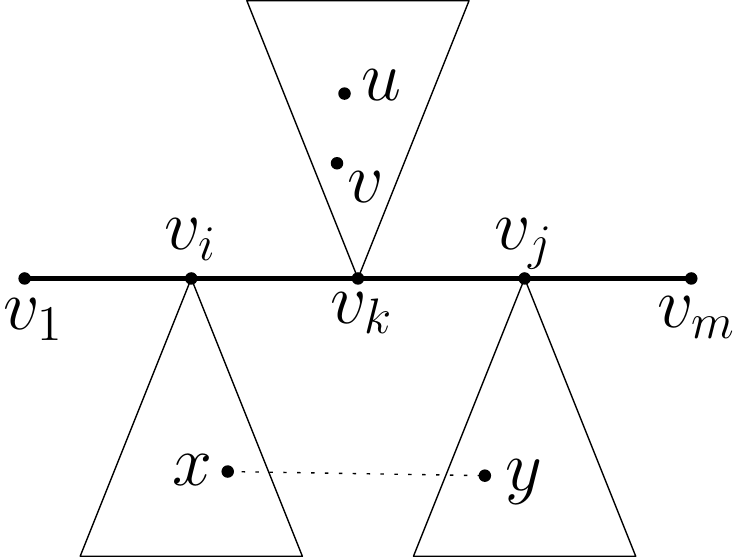}
\caption{\footnotesize Illustrating the case where both $u$ and $v$ are in $T(v_k)$.}
\label{fig:reduction70}
\end{center}
\end{minipage}
\vspace{-0.15in}
\end{figure}

If $k=h$, then both $u$ and $v$ are in $T(v_k)$; e.g., see Fig.~\ref{fig:reduction70}. Hence, it holds that $d_{T'}(u,v)\leq d_{T(v_k)}(u,v)\leq d_{T(v_k)}(v,v_k)+d_{T(v_k)}(v_k,u)$. By Lemma~\ref{lem:10}, $d_{T(v_k)}(v,v_k)\leq \min\{d_T(v_1,v_k),d_T(v_k,v_m)\}$. On the other hand, notice that the shortest path from $u$ to either $v_1$ or $v_m$ in $T'$ does not contain the shortcut $e(x,y)$. Without loss of generality, we assume that the shortest path $\pi_{T'}(u,v_1)$ from $u$ to $v_1$ in $T'$ does not contain $e(x,y)$.
Thus, $\pi_{T'}(u,v_1)$ must contain $\pi_{T}(v_1,v_k)$, implying that $\pi_{T}(v_1,v_k)$ is $\pi_{T'}(v_1,v_k)$. Note that $\pi_{T}(v_1,v_k)$ is $\pi_{P}(v_1,v_k)$, which is the portion of $P$ between $v_1$ and $v_k$. Hence, $\pi_{P}(v_1,v_k)$ is also the shortest path from $v_1$ to $v_k$ in $P'$ (i.e., $\pi_{P'}(v_1,v_k)$ does not contain the shortcut $\e(v_i,v_j)$). Therefore, $d_{P'}(v_1,v_k)=d_T(v_1,v_k)$. Since $u\in T(v_k)$, $w(v_k)\geq d_{T(v_k)}(v_k,u)$.
Combining all above we can derive
\begin{align*}
\Delta(P') &\geq w(v_1)+d_{P'}(v_1,v_k)+w(v_k) = w(v_1)+d_{T}(v_1,v_k)+w(v_k)\\
&\geq d_{T}(v_1,v_k)+w(v_k) \geq d_{T}(v_1,v_k)+d_{T(v_k)}(v_k,u)\\
&\geq d_{T(v_k)}(v,v_k)+d_{T(v_k)}(v_k,u) \geq d_{T'}(u,v).
\end{align*}

This proves the first lemma statement and thus the entire lemma.
\end{proof}

In light of Lemma~\ref{lemma:5}, we will focus on solving the DOAP problem on the vertex-weighted path $P$. Notice that the lengths of the shortcuts of $P$ have not been computed yet.


\subsection{Computing an optimal shortcut for $\boldsymbol{P}$}
\label{subsec:ComputingAnOptimalShortcutforADOAPInstance}

To find an optimal shortcut for the DOAP problem on $P$, for each $i\in [1, m-1]$, we will compute an index $j(i)$ that minimizes the diameter $\Delta(P+\e(v_i,v_j))$ among all indices $j\in[i+1, m]$, i.e., $j(i)=\arg \min_{i+1\leq j\leq m}\Delta(P+\e(v_i,v_j))$, as well as the diameter $\Delta(P+\e(v_i,v_{j(i)}))$. After that, the optimal shortcut of $P$ is the one that minimizes $\Delta(P+\e(v_i,v_{j(i)}))$ among the shortcuts $\e(i,j(i))$ for all $i\in [1,m-1]$. We refer to the shortcuts for $\e(v_i,v_j)$ for all $j\in [i+1,m]$ as {\em $v_i$-shortcuts}. Therefore, our goal is to find an optimal $v_i$-shortcut $\e(i,j(i))$ for each $i\in [1,m-1]$.

Let $n_i$ denote the number of vertices in $T(v_i)$, for each $1\leq i\leq m$. Note that $n=\sum_{i=1}^m n_i$.
Fix an index $i$ with $1\leq i\leq m-1$. In the following, we will present an algorithm that computes an optimal $v_i$-shortcut $\e(v_i,v_{j(i)})$ and the diameter $\Delta(P+\e(v_i,v_{j(i)}))$ in $O(n\cdot n_i + n\log n)$ time and $O(n)$ space. In this way, solving the DOAP problem on $P$ takes $O(n^2\log n)$ time and $O(n)$ space in total.

\begin{figure}[t]
        \centering
        \includegraphics[scale=0.4]{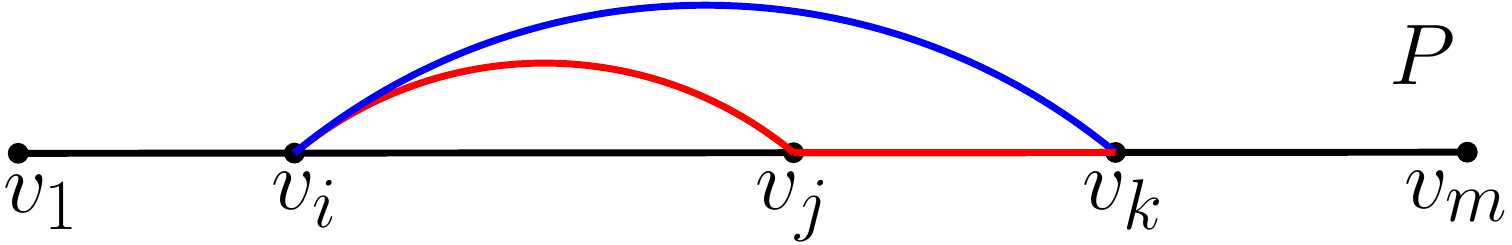}
        \caption{$\e(v_{i}, v_{j})$ dominates $\e(v_{i}, v_{k})$ if $|\e(v_{i}, v_{j})| + d_{P}(v_{j}, v_{k}) \leq |\e(v_{i}, v_{k})|$, i.e., the length of the red path is less than or equal to the length of the blue path.}
        \label{fig:VertexWeightedPath}
        \vspace{-0.15in}
\end{figure}

We introduce a domination relationship among $v_i$-shortcuts. 
\begin{definition}
For any two index $j$ and $k$ with $i < j < k \leq m$, we say that
$\e(v_{i}, v_{j})$ {\em dominates} $\e(v_{i}, v_{k})$ if $|\e(v_{i}, v_{j})| + d_{P}(v_{j}, v_{k}) \leq |\e(v_{i}, v_{k})|$; e.g., see Fig.~\ref{fig:VertexWeightedPath}.
\end{definition}

The following lemma implies that if $\e(v_i, v_j)$ dominates $\e(v_i, v_k)$, then shortcut $\e(v_i, v_k)$ can be ignored or ``pruned''.

\begin{lemma}
    \label{lemma:8}
    If $\e(v_i, v_j)$ dominates $\e(v_i, v_k)$, then $\Delta(P+\e(v_i,v_j)) \leq \Delta(P+\e(v_i,v_k))$.
\end{lemma}
\begin{proof}
    Let $u$ and $v$ be any two vertices of $P$. To prove the lemma, it suffices to show that $d_{P+\e(v_i,v_j)}(u,v)\leq d_{P+\e(v_i,v_k)}(u,v)$.

    If $d_{P+\e(v_i,v_k)}(u,v)=d_P(u,v)$, then $d_{P+\e(v_i,v_j)}(u,v)\leq d_P(u,v)=d_{P+\e(v_i,v_k)}(u,v)$.
    Otherwise, the shortest path $\pi_{P+\e(v_i,v_k)}(u,v)$ contains the shortcut $\e(v_i,v_k)$. Hence, $d_{P + \e(v_i, v_k)}(u, v)$ is equal to either $d_{P}(u, v_i) + |\e(v_i, v_k)| + d_{P}(v_k, v)$ or $d_{P}(u, v_k) + |\e(v_i, v_k)| + d_{P}(v_i, v)$. We assume that it is the former case. Then we have
    \begin{align*}
        d_{P + \e(v_i, v_k)}(u, v) &=  d_{P}(u, v_i) + |\e(v_i, v_k)| + d_{P}(v_k, v)\\
        &\geq d_{P}(u, v_i) + |\e(v_i, v_j)| + d_P(v_j,v_k) + d_{P}(v_k, v) \\
        &\geq d_{P}(u, v_i) + |\e(v_i, v_j)| + d_{P}(v_j, v)\\
        &\geq d_{P + \e(v_i, v_j)}(u, v)
    \end{align*}
   The lemma thus follows.
\end{proof}

Let $S_i$ be the set of all $v_i$-shortcuts, i.e., $S_i=\{\e(v_i,v_j) \| \ i+1\leq j\leq m\}$. In the following, we describe a pruning algorithm that computes a subset $S$ of $S_i$ such that no two shortcuts of $S$ dominate each other and $S$ contains at least one optimal $v_i$-shortcut. As will be seen later, these properties of $S$ allow an efficient algorithm to find an optimal $v_i$-shortcut.

Before running the pruning algorithm, we compute the lengths of shortcuts of $S_i$ by brute force, as follows. First, with $O(n)$ time preprocessing, given any two vertices $u$ and $v$ with $u \in T(v_i)$ and $v\in T(v_j)$ for $j\neq i$, we can compute $d_T(u,v)$ in constant time. Consider a tree $T(v_j)$ with $j\geq i+1$. Computing the length of $\e(v_i,v_j)$ reduces to finding a critical pair of $(i,j)$. To this end, we compute $d_{T(v_{i})}(v_{i}, u) + |e(u, v)| + d_{T(v_{j})}(v, v_{j})$ for all vertices $u\in T(v_i)$ and all vertices $v\in T(v_j)$, which can be done in $O(n_i\cdot n_j)$ time (and $O(n)$ space). As such, computing the lengths of all shortcuts of $S_i$ takes $O(n_i\cdot n)$ time.

Our pruning algorithm processes the shortcuts $\e(v_i,v_j)$ for all $j=i+1,i+2,\ldots, m$ one by one. A stack $S$ is maintained and $S=\emptyset$ initially. Consider any $j\in [i+1,m]$. If $S=\emptyset$, then we push $\e(v_i,v_j)$ into $S$. Otherwise, let $\e$ be the shortcut at the top of $S$. If $\e$ and $\e(v_i,v_j)$ do not dominate each other, then we push $\e(v_i,v_j)$ into $S$. Otherwise, if $\e$ dominates $\e(v_i,v_j)$, then we proceed on $j+1$, i.e., $\e(v_i,v_j)$ is pruned. If $\e(v_i,v_j)$ dominates $\e$, then we pop $\e$ out of $S$ (i.e., $\e$ is pruned). Next, we keep popping the top element out of $S$ until either $S$ becomes $\emptyset$ or $\e(v_i,v_j)$ does not dominate it; in either case we push $\e(v_i,v_j)$ into $S$.

As the lengths of the shortcuts of $S_i$ are available, the algorithm runs in $O(n)$ time. The following lemma proves the correctness of the algorithm.


\begin{lemma}
After the algorithm, no two shortcuts of $S$ dominate each other and $S$ contains at least one optimal $v_i$-shortcut.
\end{lemma}
\begin{proof}
We first prove the following observation, which will be used for proving the lemma.
\subparagraph{\em Observation.}
    {\em Consider any three indices $j$, $k$, and $h$ with $i<j<k<h\leq m$.
    \begin{enumerate}
        \item If $\e(v_i, v_{j})$ dominates $\e(v_i, v_{k})$, and $\e(v_i, v_{k})$ dominates $\e(v_i, v_{h})$, then $\e(v_i, v_{j})$ dominates $\e(v_i, v_{h})$.
        \item If $\e(v_i, v_{j})$ and $\e(v_i, v_{k})$ do not dominate each other, and $\e(v_i, v_{k})$ and $\e(v_i, v_{h})$ do not dominate each other, then $\e(v_i, v_{j})$ and $\e(v_i, v_{h})$ do not dominate each other.
    \end{enumerate}}
\subparagraph{Proof of the observation.}
Since $i<j<k<h$, $d_P(v_j,v_k)+d_P(v_k,v_h)=d_P(v_j,v_h)$.
    \begin{enumerate}
        \item If $\e(v_i, v_{j})$ dominates $\e(v_i, v_{k})$, and $\e(v_i, v_{k})$ dominates $\e(v_i, v_{h})$, then we have
        \begin{align*}
            |\e(v_i, v_{j})| + d_{P}(v_{j}, v_{h}) &= |\e(v_i, v_{j})| + d_{P}(v_{j}, v_{k}) + d_{P}(v_{k}, v_{h}) \\
            &\leq |\e(v_i, v_{k})| + d_{P}(v_{k}, v_{h}) \leq |\e(v_i, v_{h})|.
        \end{align*}
        Hence, $\e(v_i, v_{j})$ dominates $\e(v_i, v_{h})$.

        \item If $\e(v_i, v_{j})$ and $\e(v_i, v_{k})$ do not dominate each other, and $\e(v_i, v_{k})$ and $\e(v_i, v_{h})$ do not dominate each other, then we have
        \begin{align*}
            |\e(v_i, v_{j})| + d_{P}(v_{j}, v_{h}) &= |\e(v_i, v_{j})| + d_{P}(v_{j}, v_{k}) + d_{P}(v_{k}, v_{h}) \\
            & > |\e(v_i, v_{k})| + d_{P}(v_{k}, v_{h})> |\e(v_i, v_{h})|.
        \end{align*}
        Hence, $e(v_i, v_{j})$ does not dominate $e(v_i, v_{h})$.
        \begin{align*}
            |\e(v_i, v_{h})| + d_{P}(v_{j}, v_{h}) &= |\e(v_i, v_{h})| + d_{P}(v_{k}, v_{h}) + d_{P}(v_{j}, v_{k}) \\
            &> |\e(v_i, v_{k})| + d_{P}(v_{j}, v_{k}) > |\e(v_i, v_{j})|.
        \end{align*}
        Hence, $\e(v_i, v_{h})$ does not dominate $e(v_i, v_{j})$.

        Therefore, $\e(v_i, v_{j})$ and $\e(v_i, v_{h})$ do not dominate each other.
    \end{enumerate}
This proves the observation.
\medskip

We are now in a position to prove the lemma.

We first show that no two shortcuts in $S$ dominate each other by mathematical induction. Initially the statement holds, for $S=\emptyset$. We assume that the statement holds right before $\e(v_i,v_j)$ is processed.  Let $S$ refer to the stack right before $\e(v_i,v_j)$ is processed; let $\e$ be the shortcut at the top of $S$ if $S\neq \emptyset$. According to our algorithm, $\e(v_i,v_j)$ is pushed into $S$ in the following cases.
\begin{enumerate}
    \item $S=\emptyset$. In this case, $S$ has only one shortcut after $\e(v_i,v_j)$ is pushed in. Hence, the statement trivially holds.

    \item $S\neq \emptyset$, and $\e$ and $\e(v_i,v_j)$ do not dominate each other. Let $\e'$ be any shortcut in $S$. To prove the statement holds on $S$ after $\e(v_i,v_j)$ is pushed in, it suffices to show that $\e'$ and $\e(v_i,v_j)$ do not dominate each other. If $\e'=\e$, then this is obviously true. Otherwise, by the induction hypothesis, $\e'$ and $\e$ do not dominate each other. As $\e$ and $\e(v_i,v_j)$ do not dominate each other, by the above Observation, $\e'$ and $\e(v_i,v_j)$ do not dominate each other.

    \item $S\neq \emptyset$, and $\e(v_i,v_j)$ dominates $\e$. In this case, $\e$ is popped out of $S$, and afterwards, the algorithm keeps popping out the top element of $S$ until either $S$ becomes $\emptyset$ or $e(v_i,v_j)$ does not dominate the current top element of $S$, denoted by $\e''$. In the former case, the statement trivially holds. In the latter case, we claim that $\e''$ does not dominate $\e(v_i,v_j)$. Indeed, by the induction hypothesis, $\e$ and $\e''$ do not dominate each other. Since $\e(v_i,v_j)$ dominates $\e$, $\e''$ cannot dominate $\e(v_i,v_j)$ since otherwise $\e''$ would dominate $\e$ by the above Observation. Then, following the same argument as the above case, $\e(v_i,v_j)$ and $\e'$ do not dominate each other for any $e'$ in the current stack $S$. The statement thus holds after $\e(v_i,v_j)$ is pushed into $S$.
\end{enumerate}
This proves that no two shortcuts in $S$ dominate each other.

We next show that $S$ contains at least one optimal $v_i$-shortcut. Let $\e^*$ be an optimal $v_i$-shortcut. If $\e^*$ is in $S$, then we are done with the proof. Otherwise, according to our algorithm, $\e^*$ must be dominated by another shortcut $\e$ of $S_i$. By Lemma~\ref{lemma:8}, $\e$ is also an optimal $v_i$-shortcut. If $\e$ is in $S$, then we are done with the proof. Otherwise, we can inductively show that $S$ contains an optimal $v_i$-shortcut.

This completes the proof of the lemma.
\end{proof}

Using the algorithm for Theorem~\ref{theorem:1} as a subroutine,
the following lemma provides a binary search algorithm that finds an optimal $v_i$-shortcut from $S$ in $O(n\log n)$ time and $O(n)$ space.

\begin{lemma}
    \label{lemma:10}
     An optimal $v_i$-shortcut in $S$ can be found in $O(n\log n)$ time and $O(n)$ space.
\end{lemma}
\begin{proof}
We first prove some properties that our algorithm relies on.

Consider the graph $P+\e(v_i,v_j)$ for an index $j$ with $i<j\leq m$. By slightly abusing the notation, let $\Delta(i,j)=\Delta(P+\e(v_i,v_j))$. Suppose $(v_{a},v_b)$ is a diametral pair of $P+\e(v_i,v_j)$ with $a<b$. Then, $\Delta(i,j)=w(v_a)+d_{P+\e(v_i,v_j)}(v_a,v_b)+w(v_b)$.

We claim that if $a\in (1, i]$, then $(v_1,v_b)$ is also a diametral pair. Indeed, since $a\neq 1$, by Lemma~\ref{lem:10} and the definition of $w(v_a)$, we have $w(v_a)\leq d_P(v_1,v_a)$. Hence, we can derive
\begin{align*}
w(v_a) + d_{P+\e(v_i,v_j)}(v_a, v_b) + w(v_b) &\leq d_{P}(v_{1}, v_a) + d_{P+\e(v_i,v_j)}(v_a, v_b) + w(v_b)\\
& = d_{P+\e(v_i,v_j)}(v_1, v_b) + w(v_b) \\
& \leq w(v_1) + d_{P+\e(v_i,v_j)}(v_1, v_b) + w(v_b).
\end{align*}
Hence, $(v_1,v_b)$ is also a diametral pair.

Similarly, we claim that if $b\in [j,m)$, then $(v_a,v_m)$ is also a diametral pair. The claim can be proved by a similar argument as above.

Note that since $a<b$, $a\neq m$ and $b\neq 1$.
Due to the above two claims, we assume that $a \in \{1\} \cup (i, m)$ and $b \in (1, j)\cup \{m\}$.
Based on the values of $a$ and $b$, we define the following five functions (e.g., see Fig.~\ref{fig:SixFunctions}).

    \begin{figure}[t]
        \begin{subfigure}[t]{0.5\linewidth}
        \centering
        \includegraphics[scale=0.45]{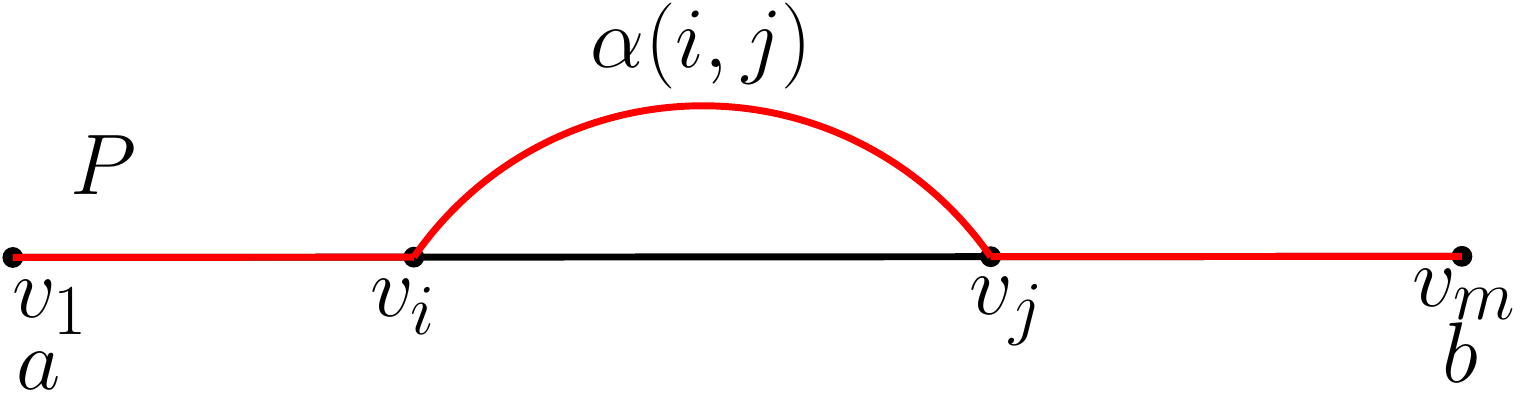}
        \caption{}
        \end{subfigure} \hfill
        \begin{subfigure}[t]{0.5\linewidth}
        \centering
        \includegraphics[scale=0.45]{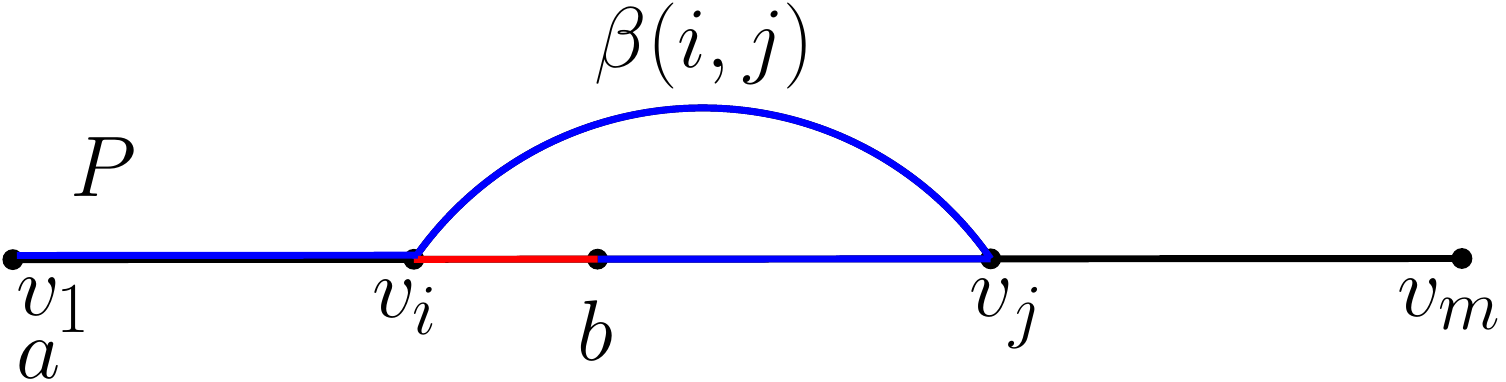}
        \caption{}
        \end{subfigure} \vfill
        \begin{subfigure}[t]{0.5\linewidth}
        \centering
        \includegraphics[scale=0.45]{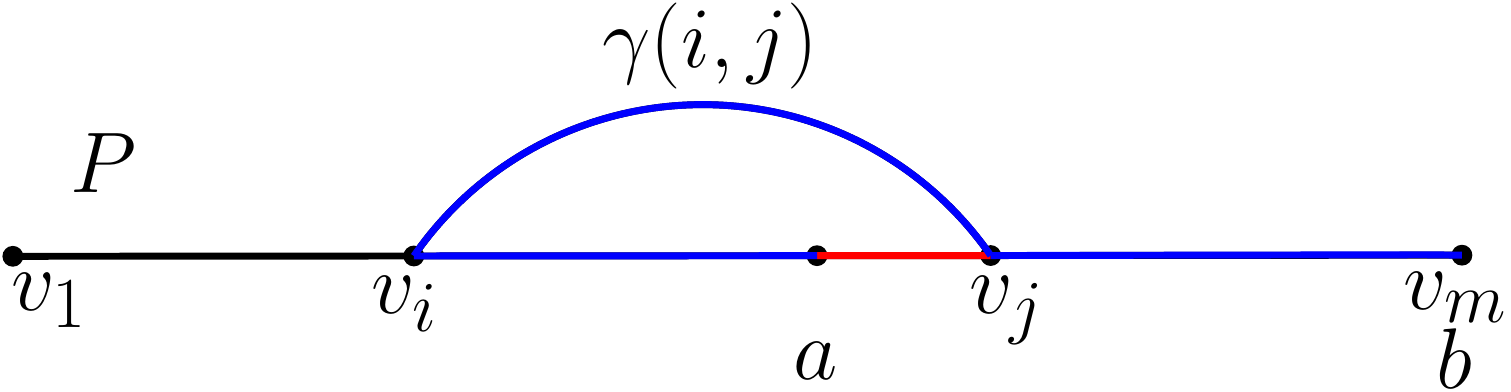}
        \caption{}
        \end{subfigure} \hfill
        \begin{subfigure}[t]{0.5\linewidth}
        \centering
        \includegraphics[scale=0.45]{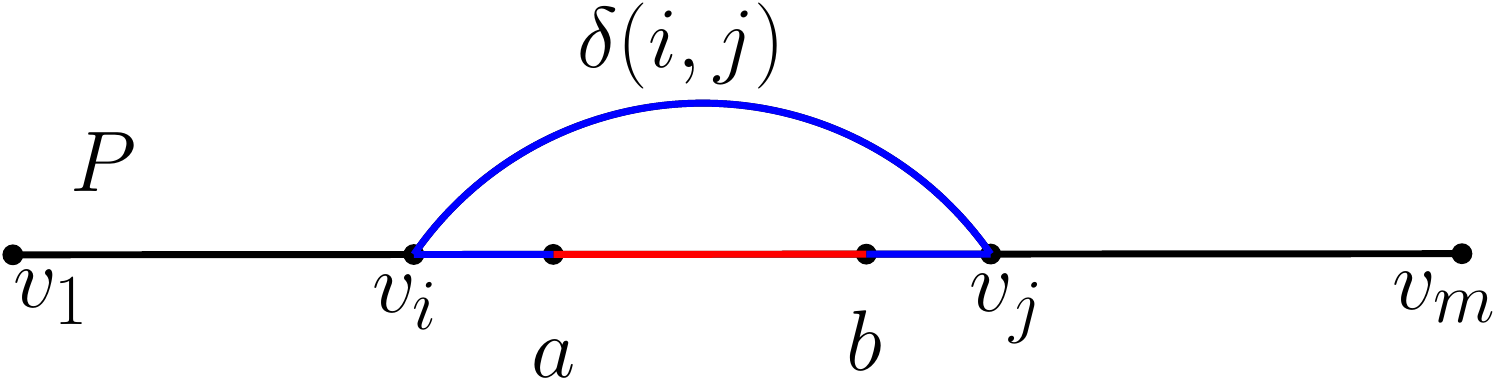}
        \caption{}
        \end{subfigure} \vfill
        \begin{subfigure}[t]{0.5\linewidth}
        \centering
        \includegraphics[scale=0.45]{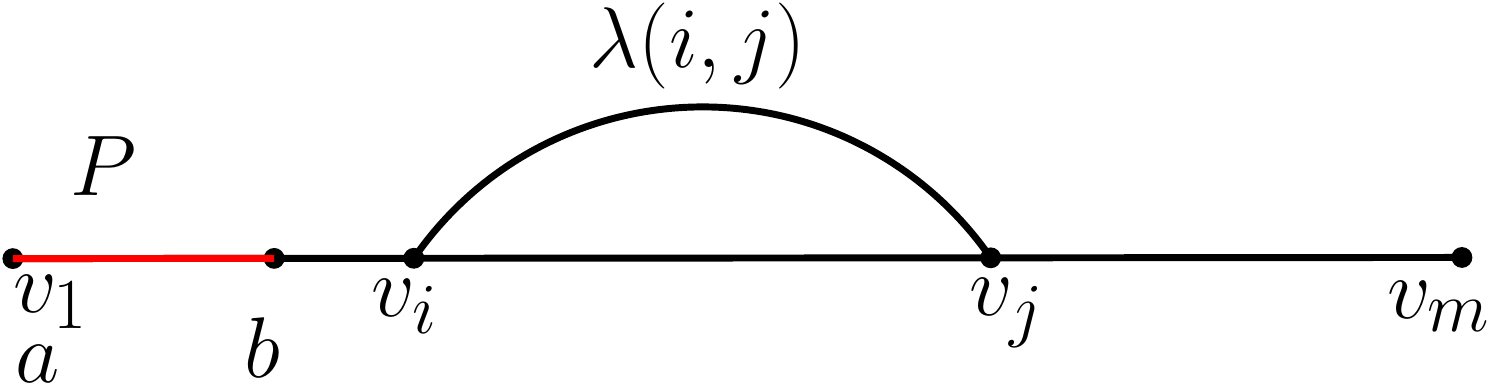}
        \caption{}
        \end{subfigure} \hfill
        \caption{(a) $\alpha(i, j)$ is the distance between $v_{1}$ and $v_{m}$. (b) $\beta(i, j)$ is the maximum distance between $v_{1}$ and all vertices on $(i, j)$. (c) $\gamma(i, j)$ is the maximum distance between $v_{m}$ and all vertices on $(i, m)$. (d) $\delta(i, j)$ is the maximum distance between two vertices on $(i, j)$. (e) $\lambda(i, j)$ is the maximum distance between $v_{1}$ and all vertices on $(1, i]$.}
        \label{fig:SixFunctions}
        \vspace{-0.15in}
    \end{figure}

\begin{enumerate}
    \item For the case $a=1$ and $b=m$, we define
    \begin{align*}
      \alpha(i, j) = w(v_1) + d_{P}(v_{1}, v_{i}) + |\e(v_i, v_j)| + d_{P}(v_{j}, v_{m})+w(v_m).
    \end{align*}
    Hence, if $a=1$ and $b=m$, we have $\Delta(i,j)=\alpha(i,j)$.

    \item For the case $a=1$ and $b\in (i,j)$, we define
    \begin{align*}
      \beta(i, j) = w(v_1)+\max_{i < b' < j} \bigg\{\min \{d_{P}(v_{1}, v_{b'}), d_{P}(v_{1}, v_{i}) + |\e(v_i, v_j)| + d_{P}(v_{j}, v_{b'}) + w(v_{b'})\}\bigg\}.
    \end{align*}
    Hence, if $a=1$ and $b\in (i,j)$, we have $\Delta(i,j)=\beta(i,j)$.


    \item For the case $a\in (i,m)$ and $b=m$, we define
    \begin{align*}
    \gamma(i, j)=\max_{i<a'<m} \bigg\{w(v_{a'})+d_{P+\e(v_i,v_j)}(v_{a'},v_m) \bigg\}+w(v_m).
    \end{align*}
    Note that $d_{P+\e(v_i,v_j)}(v_{a'},v_m)$ is equal to $\min \{d_{P}(v_{a'},v_m), d_{P}(v_{a'}, v_{i}) + |\e(v_i, v_j)| + d_{P}(v_{j}, v_{m})\}$ if $a'\in (i,j)$, and $d_{P}(v_{a'},v_m)$ otherwise.

    Hence, if $a\in (i,m)$ and $b=m$, we have $\Delta(i,j)=\gamma(i,j)$.


    \item For the case $i<a<b<j$, we define
    \begin{align*}
       \delta(i, j) = \max_{i < a' < b' < j} \bigg\{w(v_{a'}) + \min \{d_{P}(v_{a'}, v_{b'}), d_{P}(v_{a'},v_{i}) + |\e(v_i, v_j)| + d_{P}(v_{j}, v_{b'})\}+ w(v_{b'})\bigg\}.
    \end{align*}
     Hence, if $a,b\in (i,j)$, we have $\Delta(i,j)=\delta(i,j)$.


     \item For the case $a=1$ and $b\in (1,i]$, we define
     \begin{align*}
       \lambda(i, j) = \max_{1 < b' \leq i} \bigg\{w(v_{1}) + d_{P}(v_{1}, v_{b'}) + w(v_{b'})\bigg\}.
     \end{align*}
     Hence, if $a=1$ and $b\in (1,i]$, we have $\Delta(i,j)=\lambda(i,j)$.
\end{enumerate}

With the above definitions, we have $$\Delta(i, j) = \max \{\alpha(i,
j), \beta(i, j), \gamma(i, j), \delta(i, j), \lambda(i, j)\}.$$ Hence, if $j$ changes in $[i+1,m]$, the graph of $\Delta(i,j)$ is the upper envelope of the graphs of the five functions.

Recall that our goal is to find an optimal $v_i$-shortcut in $S$. Let $I$ denote the set of the indices $j$ of all shortcuts $\e(v_i,v_j)$ of $S$. We consider these indices of $I$ in order. We intend to show that $\Delta(i,j)$ is a unimodal function (first decreases and then increases) as $j$ changes in $I$. To this end, we prove that each of the above five functions is a monotonically increasing or decreasing function as $j$ changes in $I$. Note that each index of $I$ is in $[i+1,m]$. To simplify the notation, we simply let $I=\{i+1,i+2,\ldots, m\}$, or  equivalently, one may consider that our pruning algorithm does not prune any shortcut from $S_i$ and thus $S=S_i$.

As no two shortcuts of $S$ dominate each other, $\e(v_i,v_j)$ and $\e(v_i,v_{j+1})$ do not dominate each other for any $j\in (i,m)$, i.e., $|\e(v_i,v_j)|+d_P(v_j,v_{j+1})>|\e(v_i,v_{j+1})|$ and $|\e(v_i,v_{j+1})|+d_P(v_j,v_{j+1})>|\e(v_i,v_{j})|$. Relying on this property, we prove monotonicity properties of the five functions as follows. Consider any index $j\in (i,m)$.
\begin{enumerate}
    \item $\alpha(i, j) > \alpha(i, j + 1)$.

    {\em Proof: }
    \begin{align*}
    \alpha(i, j + 1) &= w(v_1) + d_{P}(v_{1}, v_{i}) + |\e(v_i, v_{j+1})| + d_{P}(v_{j+1}, v_{m}) + w(v_m) \\
    &< w(v_1)+d_{P}(v_{1}, v_{i}) + \e(v_i, v_j) + d_{P}(v_{j}, v_{j + 1}) + d_{P}(v_{j + 1}, v_{m})+w(v_m)\\
    &= w(v_1)+d_{P}(v_{1}, v_{i}) + \e(v_i, v_j) + d_{P}(v_{j}, v_{m}) +w(v_m)\\
    & = \alpha(i, j).
    \end{align*}

    Hence, $\alpha(i,j)$ is monotonically decreasing as $j$ increases.

    \item $\beta(i, j) \leq \beta(i, j + 1)$.

    {\em Proof: }
    For each $b'\in (i,j)$, we have
    \begin{align*}
        &\min \{d_{P}(v_{1}, v_{b'}), d_{P}(v_{1}, v_{i}) + |\e(v_i, v_j)| + d_{P}(v_{j}, v_{b'})\}  \\
        & \leq \min \{d_{P}(v_{1}, v_{b'}), d_{P}(v_{1}, v_{i}) + |\e(v_i, v_{j+1})| + d_P(v_{j+1},v_j)+d_{P}(v_{j}, v_{b'})\} \\
        & =\min \{d_{P}(v_{1}, v_{b'}), d_{P}(v_{1}, v_{i}) + |\e(v_i, v_{j+1})| + d_P(v_{j+1}, v_{b'})\}.
    \end{align*}
    Hence, $\beta(i, j) \leq \beta(i, j + 1)$ holds and $\beta(i,j)$ is monotonically increasing as $j$ increases.

        \item $\gamma(i, j) \geq \gamma(i, j + 1)$.

         {\em Proof: }

        Let $a'$ be the index such that $\gamma(i,j+1) = w(v_{a'})+d_{P+\e(v_i,v_{j+1})}(v_{a'},v_m)+w(v_m)$.
        Depending on whether $a'\in (i,j+1)$ or $a'\in [j+1,m)$, there are two cases.

        \begin{itemize}
            \item If $a'\in [j+1,m)$, then
            \begin{align*}
            \gamma(i, j+1) = w(v_{a'}) + d_{P}(v_{a'},v_m)+w(v_m).
            \end{align*}
            On the other hand, $d_{P+\e(v_i,v_j)}(v_{a'},v_m)=d_{P}(v_{a'},v_m)$. Hence,
            $\gamma(i, j) \geq w(v_{a'}) + d_{P}(v_{a'},v_m)+w(v_m)=\gamma(i,j+1)$.

            \item If $a'\in (i,j+1)$, then
            \begin{align*}
            \gamma(i, j+1) = w(v_{a'}) + \min \{d_{P}(v_{a'},v_m), d_{P}(v_{a'}, v_{i}) + |\e(v_i, v_{j+1})| + d_{P}(v_{j+1}, v_{m})\}+w(v_m).
            \end{align*}
            Since $\e(v_i,v_j)$ and $\e(v_i,v_{j+1})$ do not dominate each other, we have
        \begin{align*}
            & d_{P}(v_{a'},v_{i}) + |\e(v_i, v_{j+1})| + d_{P}(v_{j+1}, v_{m}) \\
            & < d_{P}(v_{a'},v_{i}) + |\e(v_i, v_j)| + d_P(v_j,v_{j+1}) + d_{P}(v_{j+1}, v_{m}) \\
            & = d_{P}(v_{a'},v_{i}) + |\e(v_i, v_j)| + d_P(v_j, v_{m}).
        \end{align*}
        Hence,
        \begin{align*}
            \gamma(i, j+1) &\leq w(v_{a'}) + \min \{d_{P}(v_{a'},v_m), d_{P}(v_{a'}, v_{i}) + |\e(v_i, v_{j})| + d_{P}(v_{j}, v_{m})\}+w(v_m).
        \end{align*}
        If $a'\in (i,j)$, then we have $\gamma(i, j) \geq w(v_{a'}) +
		\min \{d_{P}(v_{a'},v_m), d_{P}(v_{a'}, v_{i}) + |\e(v_i,
		v_{j})| + d_{P}(v_{j}, v_{m})\}+w(v_m)\geq \gamma(i,j+1)$;
		otherwise, $\gamma(i, j) \geq  w(v_{a'}) +
		d_{P}(v_{a'},v_m)+w(v_m) \geq w(v_{a'}) + \min
		\{d_{P}(v_{a'},v_m), d_{P}(v_{a'}, v_{i}) + |\e(v_i, v_{j})| +
		d_{P}(v_{j}, v_{m})\}+w(v_m)\geq \gamma(i,j+1)$. As such, in either case $\gamma(i, j) \geq \gamma(i, j + 1)$ holds.
        \end{itemize}

Hence, $\gamma(i,j)$ is monotonically decreasing as $j$ increases.


        \item $\delta(i, j) \leq \delta(i, j + 1)$.

         {\em Proof: } For each $i < a' < b' < j$, we have,
        \begin{align*}
            & \min \{d_{P}(v_{a'}, v_{b'}), d_{P}(v_{a'},v_{i}) + |\e(v_i, v_j)| + d_{P}(v_{j}, v_{b'})\}\\
            &\leq \min \{d_{P}(v_{a'}, v_{b'}), d_{P}(v_{a'},v_{i}) + |\e(v_i, v_{j+1})| + d_{P}(v_{j+1},v_j) + d_{P}(v_{j}, v_{b'})\} \\
            &\leq \min \{d_{P}(v_{a'}, v_{b'}), d_{P}(v_{a'},v_{i}) + |\e(v_i, v_{j+1})| + d_{P}(v_{j+1}, v_{b'})\}.
        \end{align*}
        Hence, $\delta(i, j) \leq \delta(i, j + 1)$ and $\delta(i,j)$
		is monotonically increasing as $j$ increases.

        \item $\lambda(i, j) = \lambda(i, j + 1)$.

         {\em Proof: } By definition, $\lambda(i,j)$ is
		 constant for all indices $j$. Thus, $\lambda(i, j) = \lambda(i, j + 1)$.

    \end{enumerate}
    \medskip

    On the basis of the above monotonicity properties of the
	functions, we present a binary search algorithm that finds an
	optimal $v_i$-shortcut in $O(n \log n)$ time and $O(n)$ space.

Our algorithm performs binary search on the indices $[l,r]$, with $l =
i + 1$ and $r = m$ initially. In each step, we decide whether we will
proceed on $[k,r]$ or on $[l,k]$, where $k = \lfloor \frac{l + r}{2}
\rfloor$. To this end, we compute $\Delta(i, k)$ and $\Delta(i, k +
1)$. By Lemma \ref{lemma:5}, $\Delta(i, k) = \Delta(T + e(x, y))$
where $(x, y)$ is a critical pair of $(i, k)$. Since $T + e(x, y)$ is a
unicycle graph, we compute $\Delta(T + e(x, y))$ in $O(n)$ time by
Theorem \ref{theorem:1}. Therefore, $\Delta(i, k)$ can be computed in
$O(n)$ time.
Note that the algorithm of Theorem \ref{theorem:1} also returns a
diametral pair for $T + e(x, y)$, and we can decide which of the five
cases for the functions $\alpha$, $\beta$, $\gamma$, $\delta$, and $\lambda$ the diametral pair belong to. We do the same for $\Delta(i, k + 1)$.
Assume that $\Delta(i, k) = f(i, k)$ and $\Delta(i, k + 1) = g(i, k +
1)$, for two functions $f$ and $g$ in $\{\alpha, \beta, \gamma,
\delta, \lambda\}$. Then we have the following cases
\begin{enumerate}
    \item $f = g$. We have the following subcases.
    \begin{itemize}
        \item $f = g \in \{\beta, \delta\}$. In this case, our
		algorithm proceeds on the interval $[l, k]$. To see this,
		since both $\beta$ and $\delta$ are monotonically increasing
		functions, we have $\Delta(i, j) \geq f(i, j) \geq f(i, k) =
		\Delta(i, k)$, for any $j \in (k, r]$.
		As such, the diameter $\Delta(i,j)$ would increase if we
		proceed on $j\in (k, r]$.

        \item $f = g \in \{\alpha, \gamma\}$. In this case, we proceed on the interval $[k, r]$ because both functions are monotonically decreasing.

        \item $f = g = \lambda$. In this case, we stop the algorithm
		and return $\e(i,k)$ as an optimal $v_i$-shortcut. To see this,
		$\Delta(i, j) \geq \lambda(i, j) = \lambda(i, k) = \Delta(i, k)$ for any $j\in [l,r]$. Hence, $\Delta(i,j)$ achieves the minimum at $j=k$ among all $j\in [l,r]$.
    \end{itemize}

    \item $f \neq g$. We have the following subcases.
    \begin{itemize}
        \item One of $f$ and $g$ is $\lambda$. In this case, by a
		similar argument as before, we return $\e(v_i,v_k)$ as an
		optimal $v_i$-shortcut if $f=\lambda$, and return
		$\e(v_i,v_{k+1})$ as an optimal $v_i$-shortcut if $g=\lambda$.

        \item $\{f,g\}=\{\beta,\delta\}$. In this case, since both
		$\beta$ and $\delta$ are increasing functions, by a similar argument as before, we proceed on the interval $[l,k]$.

        \item $\{f,g\}= \{\alpha,\gamma\}$. In this case, since both
		$\alpha$ and $\gamma$ are decreasing functions, by a similar argument as before, we proceed on the interval $[k,r]$.

        \item One of $f$ and $g$ is in $\{\beta,\delta\}$ and the
		other is in $\{\alpha,\gamma\}$. In this case, one of
		$\e(v_i,v_k)$ and $\e(v_i,v_{k+1})$ is an optimal
		$v_i$-shortcut, which can be determined by comparing $\Delta(i,k)$
		with $\Delta(i,k+1)$. To see this, without loss of generality,
		we assume that $f\in \{\beta,\delta\}$ and $g\in
		\{\alpha,\gamma\}$. Hence,
        $\Delta(i, j) \geq f(i, j) \geq f(i, k) = \Delta(i, k)$ for any $j \in [k + 1, r]$, and $\Delta(i, j) \geq g(i, j)
		\geq g(i, k+1) = \Delta(i, k+1)$ for any $j \in [l, k]$.
		As such, $\min\{\Delta(i,k),\Delta(i,k+1)\}\leq \Delta(i,j)$
		for all $j\in [l,r]$.
    \end{itemize}
\end{enumerate}

The algorithm will find an optimal $v_i$-shortcut in $O(\log n)$ iterations. As each iteration takes $O(n)$ time, the total time of the algorithm is $O(n\log n)$. The space is $O(n)$.
\end{proof}

The proof of the following theorem summarizes our algorithm for the DOAT problem.

\begin{theorem}
    \label{theorem:2}
    The DOAT problem on the tree $T$ can be solved in $O(n^2 \log n)$ time and $O(n)$ space.
\end{theorem}
\begin{proof}
We first compute a diametral path $P$ of $T$ as well as the weights $w(v_i)$ for all vertices $v_i$ of $P$, which takes $O(n)$ time. Then, for each $i\in [1,m-1]$, we compute an optimal $v_i$-shortcut $\e(v_i,v_{j(i)})$ and its diameter $\Delta(P+\e(v_i,v_{j(i)}))$ for the DOAP problem on the vertex-weighted path $P$. To do so, we first compute the critical pairs of $(i,j)$ for all $j\in [i+1,m]$ and thus the lengths of the set $S_i$ of the shortcuts $\e(v_i,v_j)$ for all $j\in [i+1,m]$. This step takes $O(n\cdot n_i)$ time and $O(n)$ space. Next, we run a pruning algorithm to prune those shortcuts that are dominated by others from $S_i$ and obtain a subset $S$ of $S_i$ such that no two shortcuts of $S$ dominate each other and $S$ contains an optimal $v_i$-shortcut; this step takes $O(n)$ time.

After having $S$, by using the algorithm of Theorem~\ref{theorem:1} as a subroutine, a binary search algorithm can find an optimal $v_i$-shortcut $\e(v_i,v_{j(i)})$ as well as the diameter $\Delta(P+\e(v_i,v_{j(i)})$. Recall that the critical pairs of $(i,j)$ for all $j\in [i+1,m]$ have been computed. Let $(x_i,y_i)$ be the critical pair of $(i,j(i))$. We store $(x_i,y_i)$ and $\Delta(P+\e(v_i,v_{j(i)})$. Other space used in this step will be disregarded when we run the same algorithm to compute the optimal $v_i$-shortcuts for other $i$'s; so the total space used in the algorithm is bounded by $O(n)$. After the optimal $v_i$-shortcuts for all $i\in [1,m-1]$ are computed, we determine the index $i$ with minimum $\Delta(P+\e(v_i,v_{j(i)})$ and return $e(x_i,y_i)$ as the optimal shortcut for the DOAT problem on $T$ and return $\Delta(P+\e(i,j(i))$ as the diameter. It takes $O(n\cdot n_i + n\log n)$ time and $O(n)$ space to compute an optimal $v_i$-shortcut for each $i\in [1,m-1]$. Hence, computing optimal $v_i$-shortcuts for all $i\in [1,m-1]$ takes $O(n^2\log n)$ time in total, for $\sum_{i=1}^m n_i=n$. The space is $O(n)$ because only constant space is occupied (for storing $(x_i,y_i)$ and $\Delta(P+\e(v_i,v_{j(i)}))$) after an optimal $v_i$-shortcut is computed for each $i$.
\end{proof}

\footnotesize
\bibliography{reference}

\end{document}